\documentclass[a4paper,11pt]{article}

\usepackage{inputenc}
\usepackage{amsmath, amsthm, amssymb}
\usepackage{color}
\usepackage{graphicx}
\usepackage{a4wide}
\usepackage{mathtools}
\usepackage{mathrsfs}
\usepackage[font=small,labelfont=it, margin=5ex]{caption}

\numberwithin{equation}{section}

\newtheorem{thm}{Theorem}[section]
\newtheorem{lem}[thm]{Lemma}
\newtheorem{prop}[thm]{Proposition}
\newtheorem{cor}[thm]{Corollary}

\theoremstyle{definition}
\newtheorem{defi}[thm]{Definition}

\newtheorem{rem}[thm]{Remark}
\newtheorem{meth}[thm]{Procedure}

\newenvironment{abst}{\begin{minipage}[c]{0.9\textwidth} \footnotesize \textbf{Abstract.}}
{\end{minipage}\\[2ex]}

\newenvironment{key}{\begin{minipage}[c]{0.9\textwidth} \footnotesize \textbf{Keywords.}}
{\end{minipage}\\[2ex]}

\newenvironment{ackno}{\begin{minipage}[c]{1\textwidth} \footnotesize \textbf{Acknowledgements.}}
{\end{minipage}\\[7ex]}

\newcommand{\eps}{\ensuremath{\varepsilon}}
 
\DeclareMathOperator{\dd}{\ensuremath\normalfont{d}}
\DeclareMathOperator{\DD}{\ensuremath\normalfont{D}}

\DeclareMathOperator{\LL}{\ensuremath\normalfont{L}}

\newcommand{\oox}{\tilde{o}^{(1)}}
\newcommand{\ooy}{\tilde{o}^{(2)}}
\newcommand{\ooz}{\tilde{o}^{(3)}}

\newcommand{\ox}{{o}^{(1)}}
\newcommand{\oy}{{o}^{(2)}}
\newcommand{\oz}{{o}^{(3)}}

\newcommand{\Fx}{\tilde{f}_1}
\newcommand{\Fy}{\tilde{f}_2}
\newcommand{\Fz}{\tilde{f}_3}
\newcommand{\fx}{\tilde{\mathbf{f}}_1}
\newcommand{\fy}{\tilde{\mathbf{f}}_2}
\newcommand{\fz}{\tilde{\mathbf{f}}_3}

\newcommand{\y}{\mathbf{y}}
\newcommand{\f}{\mathbf{f}}
\newcommand{\g}{\mathbf{g}}
\newcommand{\B}{\mathbf{B}}
\renewcommand{\dd}{\mathbf{d}}
\newcommand{\m}{\mathbf{m}}
\renewcommand{\SS}{\mathbb{S}}
\newcommand{\R}{\mathbb{R}}
\newcommand{\N}{\mathbb{N}}
\newcommand{\C}{\mathbb{C}}

\begin{document}
\begin{center}
\Large\bfseries On the Reconstruction of Dipole Directions from Spherical Magnetic Field Measurements\normalsize\mdseries
\\[3ex]
{Christian Gerhards}\footnote{University of Vienna, Computational Science Center
\\Oskar-Morgenstern-Platz 1, 1090 Vienna
\\e-mail: christian.gerhards@univie.ac.at}
\\[3ex]
\today
\end{center}

\begin{abst}
Reconstructing magnetizations from measurements of the generated magnetic potential is generally non-unique. The non-uniqueness still remains if one restricts the magnetization to those induced by an ambient magnetic dipole field (i.e., the magnetization is described by a scalar susceptibility and the dipole direction). Here, we investigate the situation under the additional constraint that the susceptibility is either spatially localized in a subregion of the sphere or that it is band-limited. If the dipole direction is known, then the susceptibility is uniquely determined under the spatial localization constraint while it is only determined up to a constant under the the assumption of band-limitedness. If the dipole direction is not known, uniqueness is lost again. However, we show that all dipole directions that could possibly generate the measured magnetic potential need to be zeros of a certain polynomial which can be computed from the given potential. We provide examples of non-uniqueness of the dipole direction and examples on how to find admissible candidates for the dipole direction under the spatial localization constraint. 
\end{abst}

\begin{key}
Inverse Magnetization Problem, Decomposition of Spherical Vector Fields, Uniqueness, Magnetic Dipoles, Susceptibility
\end{key}


\section{Introduction}

Assuming a magnetic field $\B$ of the form $\B=\nabla V$ on a sphere $\SS_R=\{x\in\R^3:|x|=R\}$ that is generated by a magnetization $\m$ on a sphere $\SS_r$ of radius $r<R$, we are interested in the question of which contributions of $\m$ can be reconstructed from knowledge of the potential $V$. In particular, we are interested in magnetizations that are induced by an ambient dipole magnetic field, i.e., $\m$ is of the form
\begin{align}\label{eqn:m}
 \m(x)=Q(x)\frac{3(x\cdot \dd)x-\dd |x|^2}{|x|^5}, \quad x\in\SS_r,
\end{align}
where $\dd\in\SS$ denotes the direction of the dipole and $Q$ the susceptibility on $\SS_r$. For brevity, all unmentioned physical quantities (such as the permeability $\mu_0$ or the actual strength of the ambient dipole magnetic field) and other constant factors are implicitly included in the function $Q$ (so, technically, $Q$ is not a susceptibility, but we still call it ' susceptibility' throughout this paper). For the magnetic field $\B$ we assume that it has no other sources than $\m$, i.e., outside $\SS_r$, it can be written in the form $\B=\nabla V$ with a harmonic potential $V$ given by
\begin{align}\label{eqn:smag0}
 V(x)=V[\m](x)=\frac{1}{4\pi}\int_{\SS_r}\m(y)\cdot\frac{x-y}{|x-y|^3} {\mathrm{d}}\omega(y),\quad x\in\mathbb{R}^3\setminus\SS_r.
\end{align}
When $\m$ is of induced form as described in \eqref{eqn:m}, we typically write $V[Q,\dd]$ instead of the more general notation $V[\m]$.

In general, even if the dipole direction $\dd$ is known, the susceptibility $Q$ is not determined uniquely by knowledge of the potential $V[Q,\dd]$ on the sphere $\SS_R$ (see, e.g., \cite{maushaak03}, where they named magnetizations that produce no magnetic potential on $\SS_R$ 'annihilators'; here, we call such magnetizations 'silent from outside'). If we make the additional assumption that $Q$ is locally supported in some subregion $\Gamma\subset\SS_{r}$, then the susceptibility is actually determined uniquely (cf. \cite{gerhards16a}, based on results from \cite{baratchart13, lima13} in a Euclidean setup). Therefore, in the latter scenario, but under the condition that the dipole direction is not known, our goal is to find suitable candidates for the dipole direction $\dd$. If the magnetization $\m$ were known, then a standard procedure such as described in \cite[Chapter 7]{butler04} can be used to derive $\dd$ from the direction of $\m$ or to see that $\m$ cannot be of the form \eqref{eqn:m}. However, just given the corresponding magnetic potential $V[\m]$ on the sphere $\SS_R$, only certain components of $\m$ can be reconstructed uniquely (cf. \cite{baratchart13,gerhards16a,gubbins11}; a summary is provided in Section \ref{sec:aux}). In other words, the question we are interested in can be reformulated as follows: Knowing only the uniquely determined components of $\m$, what can be said about $\dd$ and $Q$? An illustration of the effect of this non-uniqueness on classical methods of paleopole estimation can be found, e.g., in \cite{vervelidou16}. In the paper at hand, we investigate the influence of additional constraints on $\m$ (namely, the constraint that the magnetization is localized in a subdomain $\Gamma\subset\SS_r$ of the sphere or that it is band-limited). More precisely, we provide examples of non-uniqueness for the simultaneous reconstruction of $Q$ and $\dd$ from knowledge of $V[Q,\dd]$ on $\SS_R$, even under the mentioned additional constraints. But we also show that all possible candidates for the dipole direction $\dd$ for which the given potential can be expressed in the form $V[Q,\dd]$ need to be zeros of a particular polynomial that can be obtained from the given potential (cf. Sections \ref{sec:unique} and \ref{sec:bl}). This allows to restrict the set of candidates for the dipole direction and, to some extent, improve the handling of the non-uniqueness. 

The approach above seems to be particularly feasible for the case of the spatial localization constraint. The localization constraint could be enforced by geophysically reasonable means if one has knowledge of the true magnetization in a small subregion of $\SS_r$ or if it is known in advance that there exists a region with nearly vanishing magnetization. Being able to compute the set of admissible candidates for the dipole direction could be of use, e.g., for paleopole estimations (cf. \cite{vervelidou16} and references therein for an overview on the current procedures). The assumption that the magnetization $\m$ is concentrated on a spherical surface $\SS_r$ is fairly common in geophysical applications since magnetization typically occurs only in the upper few tens of kilometers of the Earth. Actually, for any 'sufficiently nice' volumetric magnetization in the ball $\mathbb{B}_r=\{x\in\R^3:|x|<r\}$ there can be found a magnetization concentrated on $\SS_r$ that produces the same magnetic potential on $\SS_R$, $r<R$, as its volumetric counterpart (see, e.g., \cite[Section 3]{baratchartgerhards16}). For the notion of vertically integrated magnetizations, the reader is referred to \cite{gubbins11}. Last, it should be noted that the inversion of the magnetic potential $V[\m]$ from \eqref{eqn:smag0} is closely related to the gravimetric problem (see, e.g., \cite{michel05,michel08} and references therein). However, while the gravimetric problem is unique when restricted to harmonic mass densities, the vectorial nature of the inverse magnetization problem causes the described non-uniqueness issues.

Finally, the structure of the paper at hand is as follows: In Section \ref{sec:aux}, we provide some notations and a brief recapitulation of the spherical Helmholtz and Hardy-Hodge decompositions. Latter classifies those components of the magnetization $\m$ (not necessarily of the form \eqref{eqn:m}) that are determined uniquely by knowledge of $V[\m]$ on $\SS_R$. Namely, if $\m=\tilde{\m}_1+\tilde{\m}_2+\tilde{\m}_3$ is the Hardy-Hodge decomposition, then only $\tilde{\m}_2$ is determined uniquely (e.g., \cite{baratchart13,gerhards16a,gubbins11}; we say that $\m$ and $\tilde{\m}^{(2)}$ are 'equivalent from outside'). Under the additional constraint that $\m$ is locally supported in a subdomain $\Gamma$ of the sphere $\SS_r$, both $\tilde{\m}_1$ and $\tilde{\m}_2$ are determined uniquely (cf. \cite{baratchart13,gerhards16a}). We also formulate the Helmholtz and Hardy-Hodge decompositions in terms of some well-known vector spherical harmonics, which will be of use for our considerations on band-limited magnetizations. However, it should already be noted that the constraint of $\m$ being band-limited, opposed to being spatially localized, still only yields that $\tilde{\m}_2$ is determined uniquely by $V[\m]$ on $\SS_R$.

Based on the results from Section \ref{sec:aux}, Sections \ref{sec:unique} and \ref{sec:bl} focus on the case of induced magnetizations of the form \eqref{eqn:m} under the constraint that the susceptibilities $Q$ are localized in a subregion $\Gamma\subset\SS_r$ or that $Q$ is band-limited, respectively. In both cases, we supply counter-examples to the uniqueness issue, i.e., we construct two susceptibilities $Q$ and $\overline{Q}$ and dipole directions $\dd\not=\pm\overline{\dd}$ that satisfy the respective constraints and additionally yield $V[Q,\dd]=V[\overline{Q},\overline{\dd}]$ on $\SS_R$ (throughout the course of this paper, we call $(Q,\dd)$ and $(\overline{Q},\overline{\dd})$ 'equivalent (from outside)' if they produce the same potential on $\SS_R$). Although non-uniqueness prevails under the additional constraints, for a given potential $V$ of the form \eqref{eqn:smag0}, we derive a way of computing a subset of $\SS$ which contains all dipole directions $\dd$ for which there exists a susceptibility $Q$ such that $V=V[Q,\dd]$ on $\SS_R$. Namely, in the case of spatially localized susceptibilities, the admissible dipole directions are zeros of a fourth order polynomial that can be computed from the known potential $V$ (cf. Theorem \ref{thm:reconst}). This way, we at least obtain some additional information on the otherwise non-unique problem. In the optimal case, there exists only a single zero $\pm\dd$ of the polynomial, which would guarantee uniqueness for the particular measured magnetic potential $V$ (note that uniqueness is only understood up the sign because, obviously, $V[Q,\dd]=V[-Q,-\dd]$). Similar results can be obtained for band-limited susceptibilities (cf. Section \ref{sec:bl}). However, here the degree of the polynomial of which the zeros need to be determined depends on the band-limit (furthermore, the zeros do not directly resemble the dipole direction $\dd$ but rather the vector $\mathbf{y}_\dd=(Y_{1,-1}(\dd),Y_{1,0}(\dd),Y_{1,1}(\dd))$ of spherical harmonics of degree one evaluated at the point $\dd$). Additionally, while for the spatial localization constraint, a known dipole direction uniquely determines the susceptibility, the assumption of band-limitedness only implies that a given dipole direction determines the susceptibility up to an additive constant (cf. Lemma \ref{lem:indmagbl}).

Eventually, in Section \ref{sec:num}, we provide some numerical examples on how the considerations from Section \ref{sec:unique} for spatially localized magnetizations can help to obtain suitable candidates for the dipole directions $\dd$ and on how to decide if a given potential $V$ on $\SS_R$ can be produced by a dipole induced magnetization of the form \eqref{eqn:m} in the first place. For brevity, we restrict the numerical examples to the case of spatial localization constraints (and not the constraint of band-limitation) as we believe this to be more relevant for potential applications. For notational reasons, we choose $r=1$ throughout the remainder of this paper (dipole induced magnetizations then have the form $\m(x)=Q(x)(3(x\cdot\dd)-\dd)$) while the radius $R$ of the sphere where the potential $V$ is given can still be any radius $R>1$. However, the results hold true for any $0<r<R$.

\section{Auxiliary Results and Notations}\label{sec:aux}
Throughout this paper, bold-face letters $\f,\g$ typically denote vector valued functions mapping $\SS$, $\SS_R$, or $\R^3$ into $\R^3$, while $f,g,F,G$ denote scalar valued functions mapping $\SS$, $\SS_R$, or $\R^3$ into $\R$. For brevity, we denote the unit sphere $\SS_1=\{x\in\R^3:|x|=1\}$ by $\SS$ throughout the course of this paper. Accordingly, $L^2(\SS,\R^3)$ and $H_k(\SS,\R^3)$ mean the function space of vector valued square-integrable functions and the Sobolev space as denoted, e.g., in \cite{freeden98}, respectively. $L^2(\SS)$ and $H_k(\SS)$ denote the corresponding scalar valued function spaces. For the rest of this section, we briefly recapitulate some notations and results from \cite{backus96,baratchart13,freeden98,freedenschreiner09,gerhards12,gerhards16a,gubbins11}. First, we define the following Helmholtz operators, acting at a point $x\in\SS_r$:
\begin{align}
 \ox&=\frac{x}{|x|}\,\rm{id},\label{eqn:o1}
 \\\oy&=\nabla^*,\label{eqn:o2}
 \\\oz&=\LL^*=\frac{x}{|x|}\times \nabla^*,\label{eqn:o3}
\end{align}
where $\nabla^*$ denotes the surface gradient on the unit sphere $\SS$, $\LL^*$ the surface curl gradient ($\times$ means the vector product), and $\rm{id}$ the identity operator. The Euclidean gradient is denoted by $\nabla$ and can be expressed in the form $\nabla=\frac{x}{|x|}\partial_r+\frac{1}{r}\nabla^*$, for $r=|x|$. These operators allow to decompose a spherical vector field into a radial, surface curl-free, and a surface divergence-free tangential contribution. 

\begin{thm}[Spherical Helmholtz Decomposition]\label{thm:helmdecomp}
Any function $\f\in L^2(\SS,\R^3)$ can be decomposed into
 \begin{align}
  \f=\f_1+\f_2+\f_3=\ox [f_1]+\oy [f_2]+\oz [f_3],
 \end{align}
where the scalar functions $f_1$, $f_2$, $f_3$ are uniquely determined by the conditions $\int_{\SS} f_2\,  {\mathrm{d}}\omega=\int_{\SS}f_3 \, {\mathrm{d}}\omega=0$. 
\end{thm}

A further decomposition that is of particular importance for the characterization of magnetizations is based on the spherical Hardy-Hodge operators
\begin{align}
 \oox&=\ox\left(\DD+\frac{1}{2}\right)-\oy,\label{eqn:oo1}
 \\\ooy&=\ox\left(\DD-\frac{1}{2}\right)+\oy,\label{eqn:oo2}
 \\\ooz&=\oz,\label{eqn:oo3}
\end{align}
where $\DD$ denotes the pseudo-differential operator
\begin{align}\label{eqn:d}
  \DD=\left(-\Delta^*+\frac{1}{4}\right)^{\frac{1}{2}}
\end{align}
and $\Delta^*=\nabla^*\cdot\nabla^*$ the spherical Beltrami operator. These operators above reflect the decomposition into a surface curl-free tangential contribution and two further contributions generated by the gradient of functions that are harmonic in the interior and the exterior of $\SS$, respectively. 

\begin{thm}[Spherical Hardy-Hodge Decomposition]\label{thm:shhdecomp}
Any function $\f\in {L}^2(\SS,\R^3)$ can be decomposed into
 \begin{align}\label{eqn:shhdecomp}
  \f=\fx+\fy+\fz=\oox [\Fx]+\ooy [\Fy]+\ooz [\Fz],
 \end{align}
where the scalar functions $\Fx$, $\Fy$, $\Fz$ are uniquely determined by the conditions $\int_{\SS}\Fx-\Fy\,  {\mathrm{d}}\omega=\int_{\SS} \Fz\,  {\mathrm{d}}\omega=0$. If $f_1$, $f_2$, $f_3$ are the Helmholtz scalars of $\f$ as given in Theorem \ref{thm:helmdecomp}, then
\begin{align}
 \Fx&=\frac{1}{2}\left(\DD^{-1}[f_1]-f_2+\frac{1}{2}\DD^{-1}[f_2]\right),\label{eqn:fx}
 \\\Fy&=\frac{1}{2}\left(\DD^{-1}[f_1]+f_2+\frac{1}{2}\DD^{-1}[f_2]\right),\label{eqn:fy}
 \\\Fz&=f_3.\label{eqn:fz}
\end{align}
\end{thm}

Although, the Hardy-Hodge decomposition in Theorem \ref{thm:shhdecomp} reflects the decomposition that we require to describe the uniqueness issues of the treated inverse magnetization problem, the contributions $f_1$, $f_2$, $f_3$ from the Helmholtz decomposition in Theorem \ref{thm:helmdecomp} are often easier to handle and compute (e.g., $f_1(x)=\frac{x}{|x|}\cdot \f(x)$). Therefore, the relations \eqref{eqn:fx}--\eqref{eqn:fz} can be quite helpful. Some related applications and information on such a decomposition on the Euclidean plane instead of a sphere can be found in \cite{baratchart17a,baratchart13,lima13}.

In the following, we recapitulate some earlier results on how the Hardy-Hodge decomposition characterizes the uniqueness of general magnetization $\m$ (for details and proofs, the reader is referred to \cite{baratchart13,gerhards16a}). First, we introduce the notion of equivalent magnetizations, which simply means that the two magnetizations produce the same potential $V$ (i.e., the same magnetic field $\B=\nabla V$) on some sphere $\SS_R$. In other words, if there exist two equivalent magnetizations, we have non-uniqueness (i.e., the knowledge of $V=V[\m]$ on $\SS_R$ does not uniquely determine $\m$). It should be noted that, when talking about induced magnetizations with susceptibility $Q$ and dipole direction $\dd\in\SS$, uniqueness is only meant up to the sign because, clearly, $V[Q,\dd]=V[-Q,-\dd]$. 

\begin{defi}
Two magnetizations $\m$, $\overline{\m}\in L^2(\SS,\R^3)$ are called \emph{equivalent from outside} if $V[\m]=V[\overline{\m}]$ on $\SS_R$ for an $R>1$. They are called  \emph{equivalent from inside} if $V[\m]=V[\overline{\m}]$ on $\SS_R$ for an $R<1$.  A magnetization $\m$ is called \emph{silent from outside} or \textit{inside} if it is equivalent to the zero-magnetization $\overline{\m}\equiv0$ from outside or inside, respectively (i.e., if $V[\m]\equiv0$ on $\SS_R$ for $R>1$ or $R<1$, respectively; such silent magnetizations are also frequently called annihilators).

If the magnetizations $\m$, $\overline{\m}$ are of the form \eqref{eqn:m}, with susceptibilities $Q$, $\overline{Q}$ and dipole directions $\dd$, $\overline{\dd}$, then we say that $(Q,\dd)$ and $(\overline{Q},\overline{\dd})$ are equivalent from inside/outside or we say that  $(Q,\dd)$ is silent from inside/outside if the corresponding magnetizations $\m$, $\overline{\m}$ have these properties.
\end{defi}

For us, the case $R>1$ (i.e., equivalence/silence from outside) is of major relevance since we are eventually interested in using satellite magnetic field measurements, which are obviously collected in the exterior of a planet. Now we can formulate the characterization of those contributions of $\m$ that are uniquely determined by knowledge of the potential $V[\m]$ by using the notion of equivalent magnetizations.

\begin{thm}\label{thm:sunique1}
 Let $\m\in {L}^2(\SS,\R^3)$ and its decomposition into $\tilde{\m}^{(1)}$, $\tilde{\m}^{(2)}$, $\tilde{\m}^{(3)}$ be given as in Theorem \ref{thm:shhdecomp}. Then the following assertions hold true:
 \begin{itemize}
  \item[(a)] The magnetization $\tilde{\m}^{(2)}$ is equivalent from outside to $\m$ while  $\tilde{\m}^{(1)}$ is equivalent from inside to $\m$.
  \item[(b)] The magnetization $\m$ is silent from outside if and only if $\tilde{\m}^{(2)}\equiv0$ while $\m$ is silent from inside if and only if $\tilde{\m}^{(1)}\equiv0$.
  \item[(c)] If $\textnormal{supp}(\m)\subset \Gamma$, for a region $\Gamma\subset\SS$ with $\overline{\Gamma}\not=\SS$, then $\m$ is silent from outside if and only if it is silent from inside.
  \end{itemize}
\end{thm} 

We see that the contribution $\tilde{\m}^{(2)}$ is determined uniquely by $V=V[\m]$ on a sphere $\SS_R$ of radius $R>1$. If additionally $\textnormal{supp}(\m)\subset \Gamma$, then both  $\tilde{\m}^{(1)}$  and $\tilde{\m}^{(2)}$ are determined uniquely. Observing that $3(x\cdot \dd)x-\dd |x|^2$ is non-tangential for almost all $x\in\SS$, the next corollary is a direct consequence of Theorem \ref{thm:sunique1} for dipole induced magnetizations.

\begin{cor}\label{cor:indmag}
Let $\m\in {L}^2(\SS,\R^3)$ be of the induced form \eqref{eqn:m},  with $Q\in L^2(\SS)$ and $\dd\in\SS$, and \emph{supp}$(Q)\subset\Gamma$ for a fixed region $\Gamma\subset{\SS}$ with $\overline{\Gamma}\not={\SS}$. Then there does not exist another susceptibility $\overline{Q}\in L^2(\SS)$ with supp$(\overline{Q})\subset\Gamma$ such that $(Q,\dd)$ and $(\overline{Q},\dd)$ are equivalent from outside or inside, respectively. 
\end{cor}

In other words, a spatially localized susceptibility $Q$ is uniquely determined by the knowledge of $V=V[Q,\dd]$ on a sphere $\SS_R$ of radius $R\not=1$ if $\dd$ is assumed to be given in advance. Next, we introduce two classical sets of vector spherical harmonics that reflect the decompositions from Theorems \ref{thm:helmdecomp} and \ref{thm:shhdecomp} in spectral domain. For details, the reader is referred to, e.g., \cite{backus96,edmonds57,freeden98,freedenschreiner09}.

\begin{defi}\label{def:vh}
For $n\in\mathbb{N}_0$, $k=-n,\ldots,n$, and $i=1,2,3$, we set
\begin{align*}
 \y_{n,k}^{(i)}=\big(\mu_{n}^{(i)}\big)^{-\frac{1}{2}}{o}^{(i)}Y_{n,k}
\end{align*}
and
\begin{align*}
 \tilde{\y}_{n,k}^{(i)}=\big(\tilde{\mu}_{n}^{(i)}\big)^{-\frac{1}{2}}\tilde{o}^{(i)}Y_{n,k},
\end{align*}
with normalization constants $\mu_n^{(1)}=1$, $\mu_n^{(2)}=\mu_n^{(3)}=n(n+1)$, and $\tilde{\mu}_n^{(1)}=(n+1)(2n+1)$, $\tilde{\mu}_n^{(2)}=n(2n+1)$, $\tilde{\mu}_n^{(3)}=n(n+1)$. The $Y_{n,k}$ denote an orthonormal set of scalar spherical harmonics (to be consistent with later computations in Section \ref{sec:bl}, we particularly choose $Y_{n,k}$ to be the complex-valued spherical harmonics as defined in \cite{fengler05,freedengutting13}). It is to note that the type-$(2)$ and type-$(3)$ vector spherical harmonics vanish for degree $n=0$ while this is  not the case for type (1). To avoid introducing additional notation, the type-(2) and type-(3) vector spherical harmonics should, therefore, simply be regarded as void whenever they appear for degree $n=0$.
\end{defi}

The sets $\{ \y_{n,k}^{(i)}:n\in\mathbb{N}_0,k=-n,\ldots,n,i=1,2,3\}$ and $\{\tilde{\y}_{n,k}^{(i)}:n\in\mathbb{N}_0,k=-n,\ldots,n,$ $i=1,2,3\}$ each form a complete orthonormal system in $L^2(\SS,\R^3)$.  Thus, a Fourier expansion 
\begin{align}
 \m=\sum_{i=1}^3\sum_{n=0}^\infty \sum_{k=-n}^n \big(\tilde{m}^{(i)}\big)^\wedge(n,k) \tilde{\y}^{(i)}_{n,k},
\end{align}
of a magnetization $\m$, with Fourier coefficients $(\tilde{m}^{(i)})^\wedge(n,k)=\int_{\SS}\m(y)\cdot \tilde{\y}^{(i)}_{n,k}(y)d\omega(y)$, inherits the properties of the Hardy-Hodge decomposition described in Theorem \ref{thm:sunique1}. For example., $\m$ is silent from outside if and only if all type-$(2)$ Fourier coefficients vanish, i.e.,
\begin{align}
 \big(\tilde{m}^{(2)}\big)^\wedge(n,k)=\int_{\SS}\m(y)\cdot \tilde{\y}^{(2)}_{n,k}(y)d\omega(y)=0,\quad n\geq 1,k=-n,\ldots,n.\label{eqn:fcoeff}
\end{align}
Analogously, $\m$ is silent from inside if and only if $(\tilde{m}^{(1)})^\wedge(n,k)=0$ for all $n\geq 0,k=-n$, $\ldots,n$. Just as the Helmholtz and Hardy-Hodge decomposition in Theorem \ref{thm:shhdecomp}, the two sets of vector spherical harmonics have a simple connection: obviously $\tilde{\y}^{(3)}_{n,k}={\y}^{(3)}_{n,k}$, and additionally
\begin{align}
\tilde{\y}^{(1)}_{n,k}&=\sqrt{\frac{n+1}{2n+1}}{\y}^{(1)}_{n,k}-\sqrt{\frac{n}{2n+1}}{\y}^{(2)}_{n,k},\label{eqn:ynkynkt1}
\\\tilde{\y}^{(2)}_{n,k}&=\sqrt{\frac{n}{2n+1}}{\y}^{(1)}_{n,k}+\sqrt{\frac{n+1}{2n+1}}{\y}^{(2)}_{n,k}.\label{eqn:ynkynkt2}
\end{align}

\section{Spatially Localized Induced Magnetizations}\label{sec:unique}

Let $Q,\overline{Q}\in L^2(\SS)$ with $\textnormal{supp}(Q)$, $\textnormal{supp}(\overline{Q})\subset\Gamma$ for a region $\Gamma\subset\SS$ with closure $\overline{\Gamma}\not=\SS$, and $\dd\not=\pm\overline{\dd}\in{\SS}$. In order to check whether $(Q,\dd)$ and $(\overline{Q},\overline{\dd})$ are equivalent from outside, we are lead to investigating if the residual magnetization
 \begin{align}\label{eqn:magind}
  \m^-(x)=Q(\xi)\left(3(x\cdot \dd)x-\dd\right)-\overline{Q}(x)\left(3(x\cdot\overline{\dd})x-\overline{\dd}\right),\quad x\in\SS.
 \end{align}
is silent from outside. According to Theorem \ref{thm:sunique1}, latter would imply
\begin{align}
 \tilde{m}_1^-&\equiv0,\label{eqn:cond1}
 \\\tilde{m}_2^-&\equiv0,\label{eqn:cond2}
\end{align}
which by Theorems \ref{thm:helmdecomp} and \ref{thm:shhdecomp} directly implies
\begin{align}
 m_1^-&\equiv0,\label{eqn:cond11}
 \\m_2^-&\equiv0,\label{eqn:cond22}
\end{align}
where $m_i^-$ and $\tilde{m}_i^-$, $i=1,2,3$, denote the scalar functions appearing in the  Helmholtz decomposition and the Hardy Hodge decomposition of $\m^-$ according to Theorems \ref{thm:helmdecomp} and \ref{thm:shhdecomp}, respectively. Equations \eqref{eqn:magind} and \eqref{eqn:cond11} yield
\begin{align}
 m_1^-(x)&=x\cdot \m^-(x)=2(Q(x)\dd\cdot x-\overline{Q}(x)\overline{\dd}\cdot x)=0, \quad x\in\SS,\label{eqn:m1}
\end{align}
which can be reformulated to $\overline{Q}(x)=\frac{Q(x)}{x\cdot\overline{\dd}}x\cdot \dd$ and leads to the following representation of $\m^-$:
 \begin{align}
  \m^-(x)=\frac{Q(x)}{x\cdot\overline{\dd}}\left(\overline{\dd}(x\cdot \dd)-\dd(x\cdot\overline{\dd})\right),\quad x\in\SS\setminus \{y\in \R^3:y\cdot\overline{\dd}=0\}.\label{eqn:mresrep}
\end{align}
For later reference, we define 
\begin{align}
 P_{Q,\dd,\overline{\dd}}(x)=\frac{Q(x)}{x\cdot\overline{\dd}}x\cdot \dd,\quad x\in\SS\setminus \{y\in \R^3:y\cdot\overline{\dd}=0\}.
\end{align}
Additionally, equations \eqref{eqn:cond11} and \eqref{eqn:cond22} imply that $\m^-$ has to be surface divergence-free if it is silent from outside, since it must hold $\m^-=\m_3^-$, where $\m_3^-$ is the vectorial surface divergence-free function of the Helmholtz and Hardy-Hodge decomposition of $\m^-$. Summarizing, we are lead to the following assertion on uniqueness of dipole-induced magnetizations.

\begin{lem}\label{thm:unique}
Let $Q\in L^2(\SS)$, with $\textnormal{supp}(Q)\subset\Gamma$, and $\dd\in\SS$. Then, for a given $\overline{\dd}\not=\pm\dd\in\SS$, there exists a $\overline{Q}\in L^2(\SS)$ with $\textnormal{supp}(\overline{Q})\subset\Gamma$ such that $(Q,\dd)$ and $(\overline{Q},\overline{\dd})$ are equivalent from outside if and only if $P_{Q,\dd,\overline{\dd}}\in L^2(\SS)$ and $\m^-$ as in \eqref{eqn:mresrep} is surface divergence-free. 
\end{lem}

\begin{rem}\label{rem:smooth}
In particular, the lemma above implies that if $P_{Q,\dd,\overline{\dd}}\notin L^2(\SS)$ and $\overline{\dd}\not=\pm\dd$, then there exist no other susceptibility $\overline{Q}\in L^2(\SS)$ with $\textnormal{supp}(\overline{Q})\subset\Gamma$ such that $(Q,\dd)$ and $(\overline{Q},\overline{\dd})$ are equivalent from outside. This is a condition that should guarantee uniqueness for many geophysically relevant dipole induced magnetizations as it would require the susceptibility $Q$ to be zero along a meridian.

However, in general, it is fairly easy to construct examples where non-uniqueness is given: Let $\dd\not=\pm\overline{\dd}\in\SS$ and assume $Q$ to be such that the function $P_{Q,\overline{\dd}}$ given by $P_{Q,\overline{\dd}}(x)=\frac{Q(x)}{x\cdot\overline{\dd}}$ is continuously differentiable on $\SS$. Then, in order for a $\overline{Q}$ with $\textnormal{supp}(\overline{Q})\subset\Gamma$ to exist such that $(Q,\dd)$ and $(\overline{Q},\overline{\dd})$ are equivalent from outside, Lemma \ref{thm:unique} implies that $\m^-$ as in \eqref{eqn:mresrep} has to be surface divergence-free, i.e.,
\small\begin{align}
 \left(\overline{\dd}(x\cdot \dd)-\dd(x\cdot\overline{\dd})\right)\cdot\nabla^*P_{Q,\overline{\dd}}(x)= \nabla^*\cdot\left(P_{Q,\overline{\dd}}(x)\left(\overline{\dd}(x\cdot \dd)-\dd(x\cdot\overline{\dd})\right)\right)=0,\quad x\in{\SS}.\label{eqn:ode}
\end{align}\normalsize
A closer investigation of \eqref{eqn:ode} shows that the spherical circles $\mathcal{C}_{t,\dd\times\overline{\dd}}=\{x\in{\SS}:x\cdot \frac{\dd\times\overline{\dd}}{|\dd\times\overline{\dd}|}=t\}$, $t\in[-1,1]$, represent the characteristic curves of the given differential equation and that $P_{Q,\overline{\dd}}$ has to be constant along these curves. Thus, $P_{Q,\overline{\dd}}$ has to be of the form $P_{Q,\overline{\dd}}(x)=P\big(x\cdot\frac{\dd\times\overline{\dd}}{|\dd\times\overline{\dd}|}\big)$, where $P:[-1,1]\to \R$ is a continuously differentiable function with $P(t)=0$ for all $t\in[-1,1]$ that satisfy $\mathcal{C}_{t,\dd\times\overline{\dd}}\cap({\SS}\setminus\Gamma)\not=\emptyset$. Given such a $P$, we see from Lemma \ref{thm:unique} that
 \begin{align*}
  &Q(x)=P\left(x\cdot\frac{\dd\times\overline{\dd}}{|\dd\times\overline{\dd}|}\right)x\cdot \overline{\dd},\quad  x\in{\SS},
  \\&\overline{Q}(x)=P\left(x\cdot\frac{\dd\times\overline{\dd}}{|\dd\times\overline{\dd}|}\right)x\cdot\dd,\quad x\in{\SS},
 \end{align*}
satisfy $\textnormal{supp}(Q),\textnormal{supp}(\overline{Q})\subset\Gamma$ and that $(Q,\dd)$ and $(\overline{Q},\overline{\dd})$ are equivalent from outside. An illustration of two such magnetizations, with $\Gamma\subset\SS$ being the eastern hemisphere, is shown in Figure \ref{fig:suszex}.
\end{rem}

\begin{figure}\begin{center}
 \includegraphics[scale=0.38]{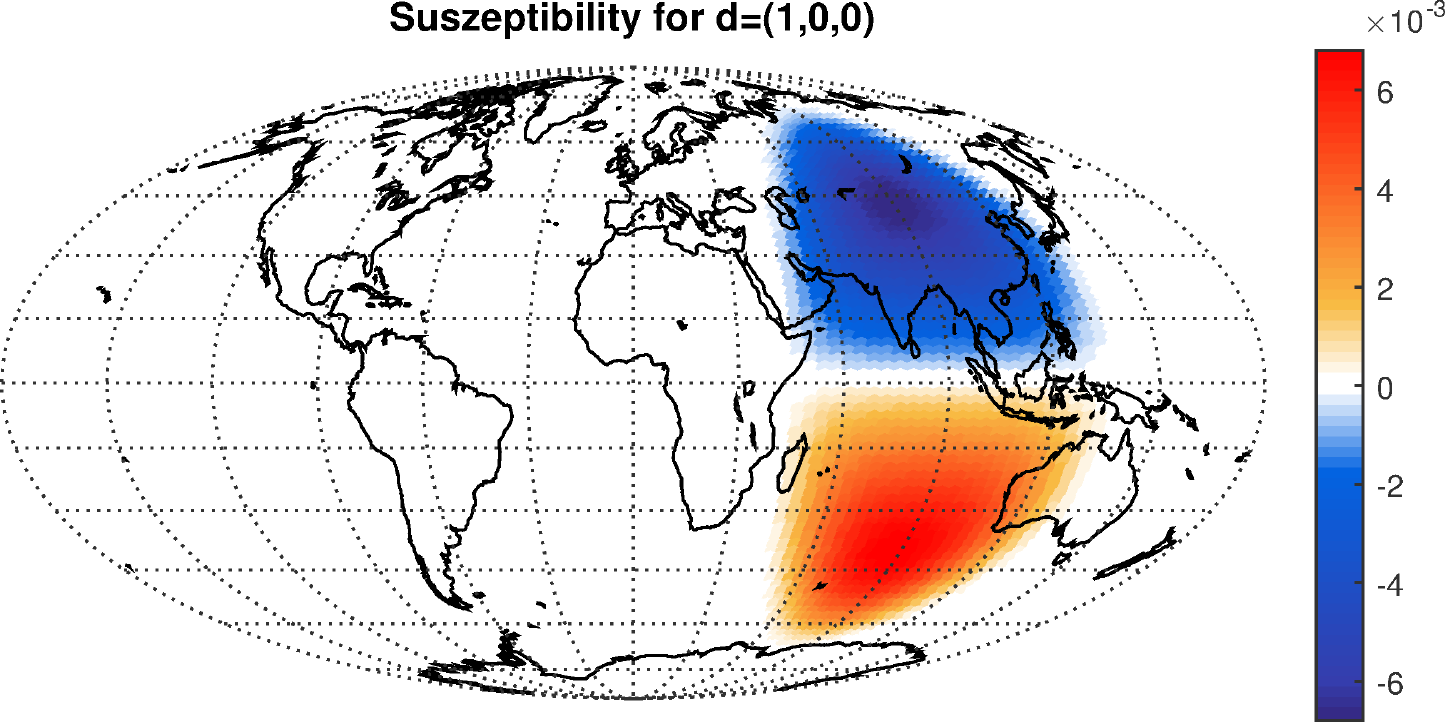}\quad\includegraphics[scale=0.38]{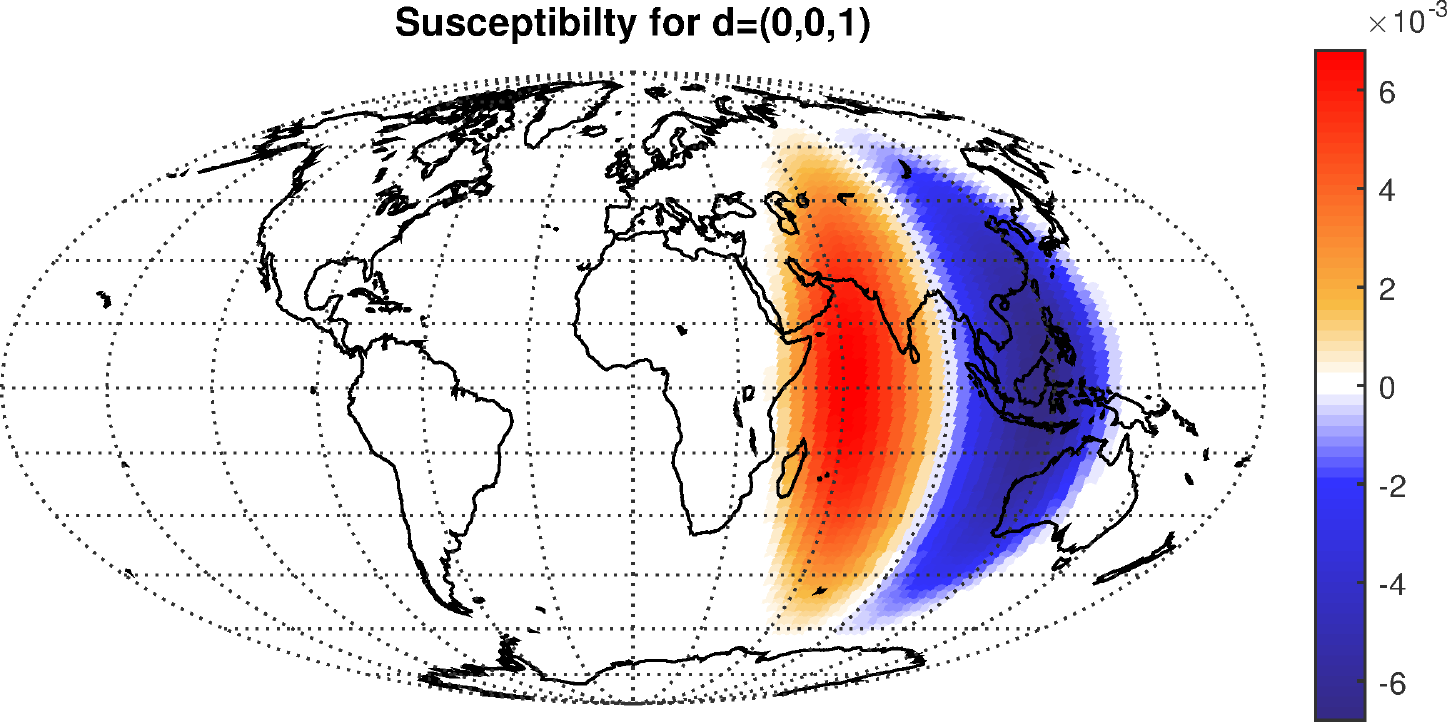}\quad\includegraphics[scale=0.38]{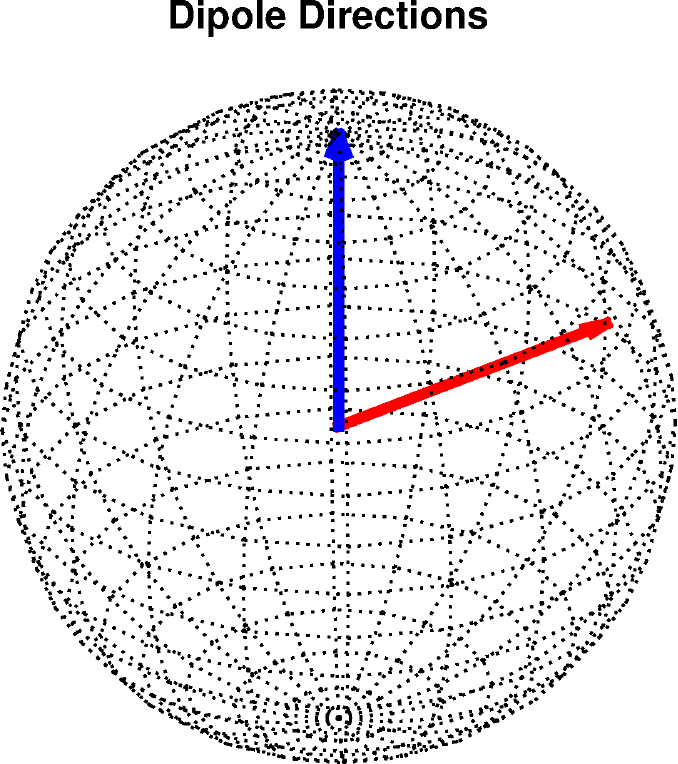} 
 \end{center}
 \caption{Illustration of two dipole induced magnetizations of the form described in Remark \ref{rem:smooth} that are equivalent from outside. We chose the auxiliary function to be $P(t)=e^{-\frac{1}{t^{2}}}\chi_{[0,1]}(t)$, the region $\Gamma=\{x\in{\SS}:(0,-1,0)^T\cdot x\geq0\}$ to be the eastern hemisphere, and the dipole directions  $\dd=(0,0,1)^T$ and $\overline{\dd}=(1,0,0)^T$, respectively. \emph{Left}: susceptibility $Q$, \emph{Center}: susceptibility $\overline{Q}$, \emph{Right}: dipole directions $\dd$ (blue) and $\overline{\dd}$ (red).}\label{fig:suszex}
\end{figure}

Let now $\overline{\m}\in H_1(\SS,\R^3)$, with $\textnormal{supp}(\overline{\m})\subset\Gamma$, and $V=V[\overline{\m}]$ be the corresponding potential on $\SS_R$. We are interested in finding out if there exists a dipole induced magnetization of the form \eqref{eqn:m} that produces the same magnetic potential on $\SS_R$ as $\overline{\m}$ (which is not necessarily of dipole induced form). If there exist $Q\in H_1(\SS)$, with $\textnormal{supp}(Q)\subset\Gamma$, and $\dd\in\SS$ such that the corresponding magnetization $\m(x)= Q(x)(3(x\cdot \dd)x-\dd)$ is equivalent from outside to $\overline{\m}$ (we also say $(Q,\dd)$ is equivalent from outside to $\overline{\m}$), then Theorem \ref{thm:sunique1} together with the Helmholtz and the Hardy-Hodge decomposition tells us
\begin{align}
 x\cdot\overline{\m}(x)=\overline{m}_1(x)=m_1(x)=2Q(x)x\cdot \dd,\quad x\in{\SS},\label{eqn:m11}
\end{align}
where $\overline{m}_1$ and $m_1$ are the radial contributions of $\overline{\m}$ and $\m$, respectively. The higher smoothness assumption of $\overline{\m},\m\in H_1(\SS,\R^3)$ is only required to allow differentiation of $\overline{\m}$ and $\m$ later on. From \eqref{eqn:m11} we get for the susceptibility $Q$ that 
\begin{align}
Q(x)=\frac{x\cdot\overline{\m}(x)}{2x\cdot \dd},\quad x\in\SS\setminus \{y\in \R^3:y\cdot\dd=0\}.\label{eqn:m12}
\end{align}
It remains to find the dipole direction $\dd\in\SS$. Again, referring to Theorem \ref{thm:sunique1} and the decompositions from Theorems \ref{thm:helmdecomp} and \ref{thm:shhdecomp}, we get, additionally to $\overline{\m}_1=\m_1$, that $\overline{\m}_2=\nabla^*\overline{m}_2=\nabla^*m_2=\m_2$. This yields, together with \eqref{eqn:m12},
\begin{align*}
 \Delta^*\overline{m}_2(x)&=\nabla^*\cdot\left(\overline{\m}(x)-\overline{\m}_1(x)\right)\nonumber
 =\nabla^*\cdot\left(Q(x)(3(x\cdot \dd)x-\dd)-\m_1(x)\right)\nonumber
 \\&=\nabla^*\cdot\left(\frac{3}{2}x\, \overline{m}_1(x)-\frac{\overline{m}_1(x)}{2(x\cdot \dd)}\dd-x\,\overline{m}_1(x)\right)\nonumber
 \\&=\overline{m}_1( x )-\frac{d}{2( x \cdot \dd)}\cdot\nabla^*\overline{m}_1( x )+\frac{\dd}{2( x \cdot \dd)^2}\cdot(\dd-( x \cdot \dd) x )\overline{m}_1( x )\nonumber
 \\&=\frac{\overline{m}_1( x )}{2}-\frac{\dd}{2( x \cdot \dd)}\cdot\nabla^*\overline{m}_1( x )+\frac{\overline{m}_1( x )}{2( x \cdot \dd)^2},\quad x \in\SS\setminus \{y\in \R^3:y\cdot\dd=0\}.
\end{align*}
Multiplying the above by $2(x\cdot \dd)^2$ leads to the condition
\begin{align}
 2(x\cdot \dd)^2\Delta^*\overline{m}_2(x)-(x\cdot \dd)^2\overline{m}_1(x)+(x\cdot \dd)\dd\cdot\nabla^*\overline{m}_1(x)-\overline{m}_1(x)=0.\label{eqn:poly2}
\end{align}
Eventually, integrating the square of the left hand side of \eqref{eqn:poly2} over $\SS$, we find that the dipole direction $\dd\in{\SS}$ has to be a zero of the following fourth-order polynomial
\small\begin{align}
 T_{\overline{g},\overline{\mathbf{h}},\overline{m}_1}(\dd)=&\sum_{i,j,k,l=1}^3 d_id_jd_kd_l\int_{\SS}  y _i y _j y _k y _l |\overline{g}( y )|^2 {\mathrm{d}}\omega( y )+ \sum_{i,j,k,l=1}^3 d_id_jd_kd_l\int_{\SS}  y _i y _j\,\overline{h}_k( y )\overline{h}_l( y )  {\mathrm{d}}\omega( y )\nonumber
 \\&+2\sum_{i,j,k,l=1}^3 d_id_jd_kd_l\int_{\SS}  y _i y _j y _k\,\overline{h}_l( y ) \overline{g}( y ) {\mathrm{d}}\omega( y )-2\sum_{i,j=1}^3 d_id_j\int_{\SS}  y _i y _j\,\overline{m}_1( y )\overline{g}( y ) {\mathrm{d}}\omega( y )\nonumber
 \\&-2\sum_{i,j=1}^3 d_id_j\int_{\SS}  y _i\,\overline{h}_j( y )\overline{m}_1( y ) {\mathrm{d}}\omega( y )+\int_{\SS}|\overline{m}_1( y )|^2 {\mathrm{d}}\omega( y ),\label{eqn:polyd}
\end{align}\normalsize
where $\overline{\mathbf{h}}=(\overline{h}_1,\overline{h}_2,\overline{h}_3)^T=\nabla^*\overline{m}_1$, $\overline{g}=2\Delta^*\overline{m}_2-\overline{m}_1$, and $\dd=(d_1,d_2,d_3)^T,y=(y_1,y_2,y_3)^T$. We can now summarize these observations in  the following theorem. 

\begin{thm}\label{thm:reconst}
Let $\overline{\m}\in H_1({\SS},\R^3)$ with $\textnormal{supp}(\overline{\m})\subset\Gamma$, for some region $\Gamma\subset{\SS}$ with $\overline{\Gamma}\not={\SS}$. Furthermore, we set $P_{\overline{\m},\dd}(x)=\frac{x\cdot\overline{\m}(x)}{2x\cdot \dd}$.  Then there exists a susceptibility $Q\in H_1({\SS})$ with $\textnormal{supp}(Q)\subset\Gamma$ and a dipole direction $\dd\in{\SS}$ such that $(Q,\dd)$ is equivalent to $\overline{\m}$ from outside if and only if there exists a $\dd\in{\SS}$ that satisfies 
\begin{align*}
T_{\overline{g},\overline{\mathbf{h}},\overline{m}_1}(\dd)=0
\end{align*}
and $P_{\overline{\m},\dd}\in H_1({\SS})$, where  $T_{\overline{g},\overline{\mathbf{h}},\overline{m}_1}$ is given as in \eqref{eqn:polyd}.
\end{thm}

\begin{rem}\label{rem:proc}
On the one hand, Theorem \ref{thm:reconst} provides a means of deciding whether a given potential $V$ on $\SS_R$ can be produced by a dipole induced magnetization of the form \eqref{eqn:m}. Namely, one first inverts $V$ to find a general magnetization $\overline{\m}$ such that $V=V[\overline{\m}]$ on $\SS_R$. Afterwards one can use Theorem \ref{thm:reconst} to check whether $V$ can also be expressed in the form $V[Q,\dd]$. On the other hand, Theorem \ref{thm:reconst} can give hints at the uniqueness of the susceptibility $Q$ and dipole direction $\dd$: if $T_{\overline{g},\overline{\mathbf{h}},\overline{m}_1}$ has only one zero (up to the sign), then uniqueness is given. 
\end{rem}


\section{Band-Limited Induced Magnetizations}\label{sec:bl}

Analogous questions as in Section \ref{sec:unique} are investigated under the assumption that the magnetization $\m$ is band-limited (and not spatially localized in the sense $\textnormal{supp}(\m)\subset\Gamma$). 

\begin{defi}\label{def:bl}
 We call a function $\f\in L^2(\SS,\R^3)$ \emph{band-limited} if there exists a $N\in \N_0$ such that
 \begin{align*}
  \big(\tilde{f}^{(i)}\big)^\wedge(n,k)=\int_{\SS}\f(y)\cdot \tilde{\y}^{(i)}_{n,k}(y)d\omega(y)=0,\quad n\geq N+1,k=-n,\ldots,n,\,i=1,2,3.,
 \end{align*}
 i.e., all Fourier coefficients vanish from some degree $N+1$ on. $N$ is called the band-limit of $\f$. A scalar function $f\in L^2(\SS)$ is called \emph{band-limited} if there exists a $N\in \N_0$ such that
 \begin{align*}
  f^\wedge(n,k)=\int_{\SS}f(y) Y_{n,k}(y)d\omega(y)=0,\quad n\geq N+1,k=-n,\ldots,n.
 \end{align*}
\end{defi}
We start by computing the Fourier expansion of magnetizations $\m$ of the form \eqref{eqn:m}. The inducing vectorial dipole field part can be expressed as
\begin{align}
 3(\dd\cdot x)x-\dd&=\sum_{k=-1}^1 \frac{8\pi}{3}Y_{1,k}(\dd)\left(\y^{(1)}_{1,k}(x)-\frac{1}{\sqrt{2}}\y^{(2)}_{1,k}(x)\right),\quad x\in\SS.
\end{align}
For a susceptibility $Q\in L^2(\SS)$ with Fourier expansion $Q=\sum_{m=0}^\infty\sum_{l=-m}^m Q^\wedge(m,l)Y_{m,l}$  one can then use the calculus of Wigner symbols (e.g., \cite{edmonds57,james73}; a notation compatible with ours is used in \cite{fengler05}) to obtain the following expression for the corresponding dipole induced magnetization $\m(x)= Q(x)\left(3(\dd\cdot x)x-\dd\right)$:
\begin{align}
 \m(x)&=\frac{8\pi}{3}\sum_{i=1}^3\sum_{p=0}^\infty \sum_{q=-p}^p\sum_{m=p-1}^{p+1}\sum_{l=q-1}^{q+1}Q^\wedge(m,l)Y_{1,k_{l,q}}(\dd)\alpha_{p,q,m,l}^{(i)}{\y}^{(i)}_{pq}(x),\quad x\in\SS,\label{eqn:mblt}
\end{align}\normalsize
where $k_{l,q}=-1$ if $l=q+1$, $k_{l,q}=0$ if $l=q$, and $k_{l,q}=1$ if $l=q-1$, and 
\begin{align}
 \alpha_{p,q,m,l}^{(1)}=&(-1)^q\sqrt{\frac{3(2m+1)(2p+1)}{4\pi}}\begin{pmatrix}
                                                          1&p&m\\0&0&0
                                                         \end{pmatrix}\begin{pmatrix}
                                                         1&p&m\\k_{l,q}&-q&l
                                                         \end{pmatrix}\label{eqn:a1}
\\\alpha_{p,q,m,l}^{(2)}=&-\frac{1}{4\sqrt{p(p+1)}}(2-m(m+1)+p(p+1))\alpha_{p,q,m,l}^{(1)}
\\\alpha_{p,q,m,l}^{(3)}=&i(-1)^{k_{l,q}}(2p+1)\sqrt{\frac{2m+1}{4\pi}}\begin{pmatrix}
                                                          p&m&1\\-q&l&k_{l,q}
                                                         \end{pmatrix}\label{eqn:a3}
                                                         \\&\times\left(\begin{Bmatrix}
                                                                                           p&p&1\\0&1&m
                                                                                          \end{Bmatrix}\begin{pmatrix}
                                                                                           p&m&0\\0&0&0
                                                                                          \end{pmatrix}
                                                                                          +\sqrt{\frac{5}{2}}\begin{Bmatrix}
                                                                                           p&p&1\\2&1&m
                                                                                          \end{Bmatrix}\begin{pmatrix}
                                                                                           p&m&2\\0&0&0
                                                                                          \end{pmatrix}\right).\nonumber
\end{align}\normalsize
The brackets $\big\{{\cdots\atop\cdots}\big\}$ denote Wigner-6j symbols while $\big({\cdots\atop\cdots}\big)$ denote Wigner-3j symbols. It is to note that round brackets are also used for matrices, however, it should be clear from the context if we mean Wigner-3j symbols or matrices. 

An expansion of the magnetization $\m$ in terms of $\tilde{\y}^{(i)}_{n,k}$, which reflects the Hardy-Hodge decomposition from Theorem \ref{thm:shhdecomp}, can be directly obtained from \eqref{eqn:mblt}, \eqref{eqn:ynkynkt1}, and \eqref{eqn:ynkynkt2}. This is summarized in the following proposition.

\begin{prop}\label{prop:mylt}
Let $Q\in L^2(\SS)$, $\dd\in\SS$, and $\alpha^{(i)}_{p,q,m,l}$, $i=1,2,3$, be the coefficients as in \eqref{eqn:a1}--\eqref{eqn:a3}. Then the dipole induced magnetization $\m(x)= Q(x)\left(3(\dd\cdot x)x-\dd\right)$ has the Fourier expansion
\begin{align*}
 \m(x)&= \frac{8\pi}{3}\sum_{i=1}^3\sum_{p=0}^\infty \sum_{q=-p}^p\sum_{m=p-1}^{p+1}\sum_{l=q-1}^{q+1}Q^\wedge(m,l)Y_{1,k_{l,q}}(\dd)\tilde{\alpha}_{p,q,m,l}^{(i)}\tilde{\y}^{(i)}_{pq}(x),\quad x\in\SS,
\end{align*}
with
\begin{align*}
\tilde{\alpha}_{p,q,m,l}^{(1)}&=\sqrt{\frac{p+1}{2p+1}}\alpha^{(1)}_{p,q,m,l}-\sqrt{\frac{p}{2p+1}}\alpha^{(2)}_{p,q,m,l}
\\\tilde{\alpha}_{p,q,m,l}^{(2)}&=\sqrt{\frac{p}{2p+1}}\alpha^{(1)}_{p,q,m,l}+\sqrt{\frac{p+1}{2p+1}}\alpha^{(2)}_{p,q,m,l}
\\\tilde{\alpha}_{p,q,m,l}^{(3)}&=\alpha^{(3)}_{p,q,m,l}.
 \end{align*}
The properties of the Wigner-3j symbols yield that $\tilde{\alpha}_{p,q,m,l}^{(1)}=\tilde{\alpha}_{p,q,m,l}^{(2)}=0$ if $p=m$, so that the fourth sum in the above representation of $\m$ has contributions only for $m\in\{p-1,p+1\}$. Any Fourier coefficients $Q^\wedge(m,l)=\int_\SS Q(y)Y_{m,l}(y)d\omega(y)$ with $l\geq m+1$ or $l\leq -m-1$ are zero by definition. 
\end{prop}

Now we are in a place to characterize silent band-limited dipole induced magnetizations. Theorem \ref{thm:sunique1}(b) essentially states that a magnetization $\m$ is silent from outside if and only if all type-(2) Fourier coefficients $(\tilde{m}^{(2)})^\wedge(p,q)$ vanish. In consequence, the representation in Proposition \ref{prop:mylt} implies that $\m$ is silent from outside if and only if 
\small\begin{align}\label{eqn:silbl}
 \sum_{l=q-1}^{q+1}Q^\wedge(p-1,l)Y_{1,k_{l,q}}(\dd)\tilde{\alpha}_{p,q,p-1,l}^{(2)}+\sum_{l=q-1}^{q+1}Q^\wedge(p+1,l)Y_{1,k_{l,q}}(\dd)\tilde{\alpha}_{p,q,p+1,l}^{(2)}=0,\quad p\geq1,q=-p,\ldots,p.
\end{align}\normalsize
For $p\geq1$, $q=-p,\ldots,p$, and $l\in\{q-1,q,q+1\}$, we can compute from the representation in Proposition \ref{prop:mylt} that
\begin{align}\label{eqn:apqpm1l}
 \tilde{\alpha}_{p,q,p-1,l}^{(2)}&=\frac{(-1)^q}{2}(p-1)\sqrt{\frac{3(2p-1)}{4\pi p}}\begin{pmatrix}
                                                                       1&p&p-1\\0&0&0
                                                                      \end{pmatrix}
                                                                      \begin{pmatrix}
                                                                       1&p&p-1\\k_{l,q}&-q&l
                                                                      \end{pmatrix},
\end{align}
so that $\tilde{\alpha}_{p,q,p-1,l}^{(2)}=0$ if and only if $p=1$. Analogously, one can see that $\tilde{\alpha}_{p,q,p+1,l}^{(2)}\not=0$ for all $p\geq 1$, $q=-p,\ldots,p$, and $l\in\{q-1,q,q+1\}$. This leads us to the following statement.

\begin{lem}\label{lem:indmagbl}
Let $Q\in L^2(\SS)$ be band-limited and $\dd\in\SS$. If  $\overline{Q}\in L^2(\SS)$ is another band-limited susceptibility such that $(Q,\dd)$ and $(\overline{Q},\dd)$ are equivalent from outside, then all Fourier coefficients for degrees greater or equal to one coincide, i.e., $Q^\wedge(m,l)=\overline{Q}^\wedge(m,l)$ for all $m\geq 1$, $l=-m,\ldots,m$. 
\end{lem}

\begin{proof}
Let us assume for now that $(Q,\dd)$ is silent from outside. The equations in \eqref{eqn:silbl} can be rewritten in the form
\begin{align}\label{eqn:Msys}
 \mathbf{M}^{\y_\dd}_{p-1}\mathbf{q}_{p-1}+\mathbf{N}^{\y_\dd}_{p+1}\mathbf{q}_{p+1}=0,\quad p\geq 1,
\end{align}
where $ \mathbf{M}^{\y_\dd}_{p-1}\in\C^{(2p+1)\times (2p-1)}$, $\mathbf{N}^{\y_\dd}_{p+1}\in\C^{(2p+1)\times (2p+3)}$ are tri-band matrices and $\mathbf{q}_{p-1}\in \C^{2p-1}$, $\mathbf{q}_{p+1}\in \C^{2p+3}$ vectors. More precisely, the matrix $\mathbf{M}^{\y_\dd}_{p-1}$ and the vector $\mathbf{q}_{p-1}$ are of the form
\begin{align}\label{eqn:M}
 \mathbf{M}^{\y_\dd}_{p-1}=\begin{pmatrix}
                   a_{-p}&0&\cdots&0
                   \\b_{-p+1}&a_{-p+1}&\ddots&\vdots
                   \\c_{-p+2}&b_{-p+2}&\ddots&0
                   \\0&\ddots&\ddots&a_{p-2}
                   \\\vdots&\ddots&\ddots&b_{p-1}
                   \\0&\cdots&0&c_{p}
                  \end{pmatrix},\quad
                  \mathbf{q}_{p-1=}\begin{pmatrix}
                                   Q^\wedge(p-1,-(p-1))\\\vdots\\Q^\wedge(p-1,p-1)
                                  \end{pmatrix},
\end{align}
with $a_q=\tilde{\alpha}^{(2)}_{p,q,p-1,q+1}y_{\dd,1}$, $b_q=\tilde{\alpha}^{(2)}_{p,q,p-1,q}y_{\dd,2}$, $c_q=\tilde{\alpha}^{(2)}_{p,q,p-1,q-1}y_{\dd,3}$, and the auxiliary vector $\y_\dd=(y_{\dd,1},y_{\dd,2},y_{\dd,3})=(Y_{1,-1}(\dd),Y_{1,0}(\dd),Y_{1,1}(\dd))$.
The matrix $\mathbf{N}^{\y_\dd}_{p+1}$ has the form
\begin{align}\label{eqn:M2}
 \mathbf{N}^{\y_\dd}_{p+1}=\begin{pmatrix}
                   \gamma_{-p}&\beta_{-p}&\alpha_{-p}&0&\cdots&0
                   \\0&\gamma_{-p+1}&\beta_{-p+1}&\alpha_{-p+1}&\ddots&\vdots
                   \\\vdots&\ddots&\ddots&\ddots&\ddots&0
                   \\0&\cdots&0&\gamma_{p}&\beta_p&\alpha_p
                  \end{pmatrix},\quad
                  \mathbf{q}_{p+1=}\begin{pmatrix}
                                   Q^\wedge(p+1,-(p+1))\\\vdots\\Q^\wedge(p+1,p+1)
                                  \end{pmatrix},
\end{align}
with $\alpha_q=\tilde{\alpha}^{(2)}_{p,q,p+1,q+1}y_{\dd,1}$, $\beta_q=\tilde{\alpha}^{(2)}_{p,q,p+1,q}y_{\dd,2}$, $\gamma_q=\tilde{\alpha}^{(2)}_{p,q,p+1,q-1}y_{\dd,3}$. For $p\geq 1$ we have seen in \eqref{eqn:apqpm1l} that $\tilde{\alpha}_{p,q,p+1,l}^{(2)}\not=0$ and for $p\geq 2$ that $\tilde{\alpha}_{p,q,p-1,l}^{(2)}\not=0$. Furthermore, for any $\dd\in\SS$, at least one of the expressions $Y_{1,-1}(\dd)$, $Y_{1,0}(\dd)$, $Y_{1,1}(\dd)$ is non-zero. In consequence, for $p\geq 2$, the entries of at least one of the three main diagonals of $\mathbf{M}^{\y_\dd}_{p-1}$ are all non-zero, so that the matrix has full rank, i.e., $\textnormal{rank}(\mathbf{M}^{\y_\dd}_{p-1})=2p-1$. The same holds true for $\mathbf{N}^{\y_\dd}_{p+1}$.

Since $Q$ is band-limited, there must exist a $N\in\N_0$ such that $\mathbf{q}_{p+1}=0$, for $p\geq N$. Thus, iteratively, we obtain from \eqref{eqn:Msys} and the full rank of $\mathbf{M}^{\y_\dd}_{p-1}$ that any band-limited dipole induced magnetization that is silent from outside has to satisfy $\mathbf{q}_{p-1}=0$, for $p\geq 2$, i.e., $Q^\wedge(m,l)=0$ for all $m\geq 1$, $l=-m,\ldots,m$.

If $(Q,\dd)$ is not silent from outside but $Q$ and $\overline{Q}$ are two susceptibilities such that $(Q,\dd)$ and $(\overline{Q},\dd)$ are equivalent from outside, then the difference of the two corresponding magnetizations must be silent from outside, i.e., it must be satisfied that
\begin{align}
 \mathbf{M}^{\y_\dd}_{p-1}(\mathbf{q}_{p-1}-\overline{\mathbf{q}}_{p-1})+\mathbf{N}^{\y_\dd}_{p+1}(\mathbf{q}_{p+1}-\overline{\mathbf{q}}_{p+1})=0,
\end{align}
for all $p\geq 1$. Now, the previous considerations imply the statement of the lemma. 
\end{proof} 

\begin{rem}\label{rem:nullbl}
 Equation \eqref{eqn:silbl} contains contributions of the Fourier coefficient $Q^\wedge(0,0)$ only for the choice $p=1$. The observations in  \eqref{eqn:apqpm1l}, however, yield that $\tilde{\alpha}_{1,q,0,0}^{(2)}=0$, $q\in\{-1,0,1\}$, so that $Q^\wedge(0,0)$ does not have any effect on the magnetic potential $V[Q,\dd]$ on $\SS_R$, $R>1$. In other words, any constant susceptibility $Q$ leads to a dipole induced magnetization that is silent from outside. Lemma \ref{lem:indmagbl} implies that those are all silent band-limited dipole induced magnetizations.
 
If the particular dipole direction $\dd=(0,0,1)$ is chosen, then $Y_{1,k_{l,q}}(\dd)=0$ for $k_{l,q}\not=0$. For this setting, the equations \eqref{eqn:silbl} reduce to
\small\begin{align}\label{eqn:dneff}
 Q^\wedge(p-1,q)Y_{1,0}(\dd)\tilde{\alpha}_{p,q,p-1,q}^{(2)}+Q^\wedge(p+1,q)Y_{1,0}(\dd)\tilde{\alpha}_{p,q,p+1,q}^{(2)}=0,\quad p\in\N_0,q=-p,\ldots,p.
\end{align}\normalsize
Latter is essentially identical to the recursion relation that was obtained in \cite{maushaak03} to characterize silent magnetizations (which they called annihilators). In this sense, Lemma \ref{lem:indmagbl} and the first part of this remark are just slightly more general statements of these results.
\end{rem}

Next, we are interested in the equivalence of two dipole induced magnetizations with possibly different dipole directions. More precisely, for a given band-limited $Q\in L^2(\SS)$ and $\dd\in\SS$, we want to determine if there exists another susceptibility $\overline{Q}\in L^2(\SS)$ and dipole direction $\overline{\dd}\not=\pm\dd\in\SS$ such that  $(Q,\dd)$ and $(\overline{Q},\overline{\dd})$ are equivalent from outside (for $\overline{\dd}=\pm\dd$ this is, of course, always possible by Lemma \ref{lem:indmagbl} and Remark \ref{rem:nullbl}). Equivalence from outside means that the residual magnetization
\begin{align}\label{eqn:magind2}
  \m^-(x)=Q(\xi)\left(3(x\cdot \dd)x-\dd\right)-\overline{Q}(x)\left(3(x\cdot\overline{\dd})x-\overline{\dd}\right),\quad x\in\SS.
 \end{align}
 is silent from outside. According to \eqref{eqn:silbl} and \eqref{eqn:Msys} this is possible if and only if
\begin{align}\label{eqn:Msys2}
 \mathbf{M}^{\y_\dd}_{p-1}\mathbf{q}_{p-1}+\mathbf{N}^{\y_\dd}_{p+1}\mathbf{q}_{p+1}-\mathbf{M}^{\y_{\overline{\dd}}}_{p-1}\overline{\mathbf{q}}_{p-1}-\mathbf{N}^{\y_{\overline{\dd}}}_{p+1}\overline{\mathbf{q}}_{p+1}=0,\quad p\geq1.
\end{align}
The quantities $\mathbf{M}^{\y_\dd}_{p-1}$, $\mathbf{N}^{\y_\dd}_{p+1}$, $\mathbf{q}_{p\pm1}$ are defined as in \eqref{eqn:M} and \eqref{eqn:M2}. $\overline{\mathbf{q}}_{p\pm1}$ denotes the counterpart of $\mathbf{q}_{p\pm1}$ corresponding to $\overline{Q}$. Since  the susceptibilities $Q$, $\overline{Q}$ are assumed to be band-limited, there exists some $N\in\N_0$ such that $\mathbf{q}_{p+1}=\overline{\mathbf{q}}_{p+1}=0$ for all $p\geq N$, so that, for $p\in\{N,N+1\}$, \eqref{eqn:Msys2} reduces to 
\begin{align}\label{eqn:Msys3}
 \mathbf{M}^{\y_\dd}_{p-1}\mathbf{q}_{p-1}-\mathbf{M}^{\y_{\overline{\dd}}}_{p-1}\overline{\mathbf{q}}_{p-1}=0.
\end{align}
Now, given $\mathbf{q}_{p-1}$ and $\dd$, the first question to answer is if there exist $\overline{\mathbf{q}}_{p-1}\in\C^{2p-1}$ and a dipole direction $\overline{\dd}\not=\pm\dd\in\SS$ such that \eqref{eqn:Msys3} is satisfied. The system of linear equations is overdetermined, but from the proof of Lemma \ref{lem:indmagbl} we know that $\mathbf{M}^{\y_\dd}_{p-1}$ has full rank. From now on, we assume that $\overline{\dd}\in\SS\setminus\{(0,0,\pm1)\}$ because then $Y_{1,1}(\overline{\dd})\not=0$. This yields that the matrix $\widehat{\mathbf{M}}^{\y_{\overline{\dd}}}_{p-1}$, which is obtained from $\mathbf{M}^{\y_{\overline{\dd}}}_{p-1}$ by deleting the first two rows, is invertible. Analogously, $\widehat{\mathbf{M}}^{\y_\dd}_{p-1}$ denotes ${\mathbf{M}}^{\y_\dd}_{p-1}$ with its first to rows deleted. The uniquely determined candidate for $\overline{\mathbf{q}}_{p-1}\in\C^{2p-1}$, $p\in\{N,N+1\}$, is then obtained by
\begin{align}\label{eqn:solq}
 \overline{\mathbf{q}}_{p-1}=\left(\widehat{\mathbf{M}}^{\y_{\overline{\dd}}}_{p-1}\right)^{-1}\widehat{\mathbf{M}}^{\y_\dd}_{p-1}\mathbf{q}_{p-1}.
\end{align}
It remains to check whether \eqref{eqn:Msys3} is valid for this $\overline{\mathbf{q}}_{p-1}$, i.e., if 
\begin{align}\label{eqn:Msys4}
 \mathbf{M}^{\y_\dd}_{p-1}\mathbf{q}_{p-1}-\mathbf{M}^{\y_{\overline{\dd}}}_{p-1}\left(\widehat{\mathbf{M}}^{\y_{\overline{\dd}}}_{p-1}\right)^{-1}\widehat{\mathbf{M}}^{\y_\dd}_{p-1}\mathbf{q}_{p-1}=0
\end{align}
holds true for $p\in\{N,N+1\}$. By construction it actually suffices to check if the first two rows of the above system of equations hold true. From the structure of the inverse of upper triangular matrices, we find that 
$(\Pi_{q=-p+2}^p{c}_q(p,\y_{\overline{\dd}}))\big(\widehat{\mathbf{M}}^{\y_{\overline{\dd}}}_{p-1}\big)^{-1}$ is a matrix with entries that are polynomials with respect to $\mathbf{y}_{\overline{\dd}}=(Y_{1,-1}(\overline{\dd}),Y_{1,0}(\overline{\dd}),Y_{1,1}(\overline{\dd}))\in\C^3$ (the coefficients ${c}_q={c}_q(p,\y_{\overline{\dd}})$ are defined as in \eqref{eqn:M} and depend on $p$ and $\y_{\overline{\dd}}$). In conclusion, if $\overline{\dd}\in\SS\setminus\{(0,0,\pm1)\}$ is an admissible candidate for a dipole direction, then $\y_{\overline{\dd}}$ has to be a zero of the vector-valued polynomials $\mathbf{R}_{p,Q,\dd}$ given by
\small\begin{align}\label{eqn:blpoly1}
 \mathbf{R}_{p,Q,\dd}(\y)=(\Pi_{q=-p+2}^p{c}_q(p,\y))\left(\mathbf{M}^{\y_\dd}_{p-1}\mathbf{q}_{p-1}-\mathbf{M}^{\y}_{p-1}\left(\widehat{\mathbf{M}}^{\y}_{p-1}\right)^{-1}\widehat{\mathbf{M}}^{\y_\dd}_{p-1}\mathbf{q}_{p-1}\right),\quad \y\in\C^3, 
\end{align}\normalsize
for $p\in\{N,N+1\}$. It remains to check the cases $p=2,\ldots,N-1$ ($p=1$ is not of interest since the coefficients $Q^\wedge(0,0)$ and $\overline{Q}^\wedge(0,0)$ can be chosen arbitrarily according to Lemma \ref{lem:indmagbl} and Remark \ref{rem:nullbl}). In this case, the second and fourth summand in \eqref{eqn:Msys2} cannot be omitted and we get that $\mathbf{y}_{\overline{\dd}}$ additionally needs to be a zero of the polynomials $\mathbf{S}_{p,Q,\dd}$ given by
\small\begin{align}\label{eqn:blpoly2}
 \mathbf{S}_{p,Q,\dd}(\y)=(\Pi_{n=p}^{N}\Pi_{q=-n+2}^n{c}_q(n,\y))&\bigg(\mathbf{M}^{\y_\dd}_{p-1}\mathbf{q}_{p-1}+\mathbf{N}^{\y_\dd}_{p+1}\mathbf{q}_{p+1}-\mathbf{N}^{\y}_{p+1}\overline{\mathbf{q}}_{p+1}
 \\&-\mathbf{M}^{\y}_{p-1}\left(\widehat{\mathbf{M}}^{\y}_{p-1}\right)^{-1}\left(\widehat{\mathbf{M}}^{\y_\dd}_{p-1}\mathbf{q}_{p-1}+\widehat{\mathbf{N}}^\dd_{p+1}\mathbf{q}_{p+1}-\widehat{\mathbf{N}}^{\y}_{p+1}\overline{\mathbf{q}}_{p+1}\right)\bigg),\nonumber
\end{align}\normalsize
for $p=2,\ldots,N-1$. The additional product $\Pi_{n=p}^{N}$ is only included to guarantee that $\mathbf{S}_{p,Q,\dd}$ is a polynomial, although this is not crucial for our statements. The required vectors $\overline{\mathbf{q}}_{p+1}$ can be computed iteratively from the results of the previous steps: starting with  \eqref{eqn:solq} for $p\in\{N,N+1\}$ and continuing with
\begin{align}\label{eqn:solq2}
 \overline{\mathbf{q}}_{p-1}=\left(\widehat{\mathbf{M}}^{\y_{\overline{\dd}}}_{p-1}\right)^{-1}\left(\widehat{\mathbf{M}}^{\y_\dd}_{p-1}\mathbf{q}_{p-1}+\widehat{\mathbf{N}}^\dd_{p+1}\mathbf{q}_{p+1}-\widehat{\mathbf{N}}^{\y_{\overline{\dd}}}_{p+1}\overline{\mathbf{q}}_{p+1}\right),
\end{align}
for $p=N-1,\ldots,2$.

Eventually, we see that in order to determine if, for a given band-limited susceptibility $Q$ and dipole direction $\dd\in\SS$, there exists another band-limited susceptibility $\overline{Q}$ and dipole direction $\overline{\dd}\in\SS\setminus\{(0,0,\pm1)\}$ such that $(Q,\dd)$ and $(\overline{Q},\overline{\dd})$ are equivalent from outside, one possible way is to find common zeros of $\mathbf{R}_{p,Q,\dd}$, $p\in\{N,N+1\}$, and  $\mathbf{S}_{p,Q,\dd}$, $p=2,\ldots, N-1$. If a common zero $\y\in\C^3$ other than $\y=\y_{\pm\dd}$ exists and if it is of the form $\y=\y_{\overline{\dd}}=(Y_{1,-1}(\overline{\dd}),Y_{1,0}(\overline{\dd}),Y_{1,1}(\overline{\dd}))$, then a candidate for $\overline{\dd}$ has been found (and the corresponding suceptibility $\overline{Q}$ is determined up to a constant via the Fourier coefficients gathered in  \eqref{eqn:solq}, \eqref{eqn:solq2}). However, it is by no means true that all common zeros of $\mathbf{R}_{p,Q,\dd}$ and $\mathbf{S}_{p,Q,\dd}$ need to be representable in the form $(Y_{1,-1}(\overline{\dd}),Y_{1,0}(\overline{\dd}),Y_{1,1}(\overline{\dd}))$ in the first place.  Finally, the so far excluded case $\overline{\dd}=(0,0,\pm1)$ has to be checked separately (e.g., by choosing $\widehat{\mathbf{M}}^{\y_{\overline{\dd}}}_{p-1}$ to be the matrix that is obtained from ${\mathbf{M}}^{\y_{\overline{\dd}}}_{p-1}$ not by deleting the first two rows but by deleting the first and last row). 

\begin{rem}\label{rem:suszbl}
 From Lemma \ref{lem:indmagbl} and Remark \ref{rem:nullbl} it is clear that for constant susceptibilities $Q\equiv c$ and $\overline{Q}\equiv{c}$, with $c,{c}\not=0$, it holds that $(Q,\dd)$ and $(\overline{Q},\overline{\dd})$ are equivalent from outside for any $\dd,\overline{\dd}\in\SS$. A slightly more complex example for equivalent band-limited magnetizations would be for band-limit $N=1$. Let us choose $\dd=(0,0,1)$ and $\overline{\dd}=(1,0,0)$ and construct $\mathbf{q}_1$ and $\overline{\mathbf{q}}_1$ from \eqref{eqn:solq} and \eqref{eqn:Msys4}. Clearly, $\mathbf{q}_1$ needs to be in the nullspace of of the matrix $\mathbf{M}^{\y_\dd}_{1}-\mathbf{M}^{\y_{\overline{\dd}}}_{1}\left(\widehat{\mathbf{M}}^{\y_{\overline{\dd}}}_{1}\right)^{-1}\widehat{\mathbf{M}}^{\y_\dd}_{1}$, which is spanned by $\mathbf{q}_1=\big(-\frac{1}{\sqrt{2}},0,-\frac{1}{\sqrt{2}}\big)$. From \eqref{eqn:solq} we then obtain $\overline{\mathbf{q}}_1=(0,1,0)$. This leads us to band-limited susceptibilities 
 \begin{align}
  Q(x)&=-\frac{1}{\sqrt{2}}Y_{1,-1}(x)-\frac{1}{\sqrt{2}}Y_{1,1}(x),\quad x\in \SS,
  \\\overline{Q}(x)&=Y_{1,0}(x),\quad x\in \SS.
 \end{align}
 They are illustrated in Figure \ref{fig:suszex2}. We see that, by the procedure described in the previous paragraphs, it is easy to construct band-limited $(Q,\dd)$ and $(\overline{Q},\overline{\dd})$, with $\dd\not=\pm\overline{\dd}$, that are equivalent from outside. However, to check if, for a given $(Q,\dd)$, there exists another band-limited $\overline{Q}$ and $\overline{\dd}\not=\pm\dd$ such that $(Q,\dd)$ and $(\overline{Q},\overline{\dd})$ are equivalent from outside is somewhat more tedious. But essentially it boils down to finding zeros of polynomials.
\end{rem}

 \begin{figure}\begin{center}
 \includegraphics[scale=0.44]{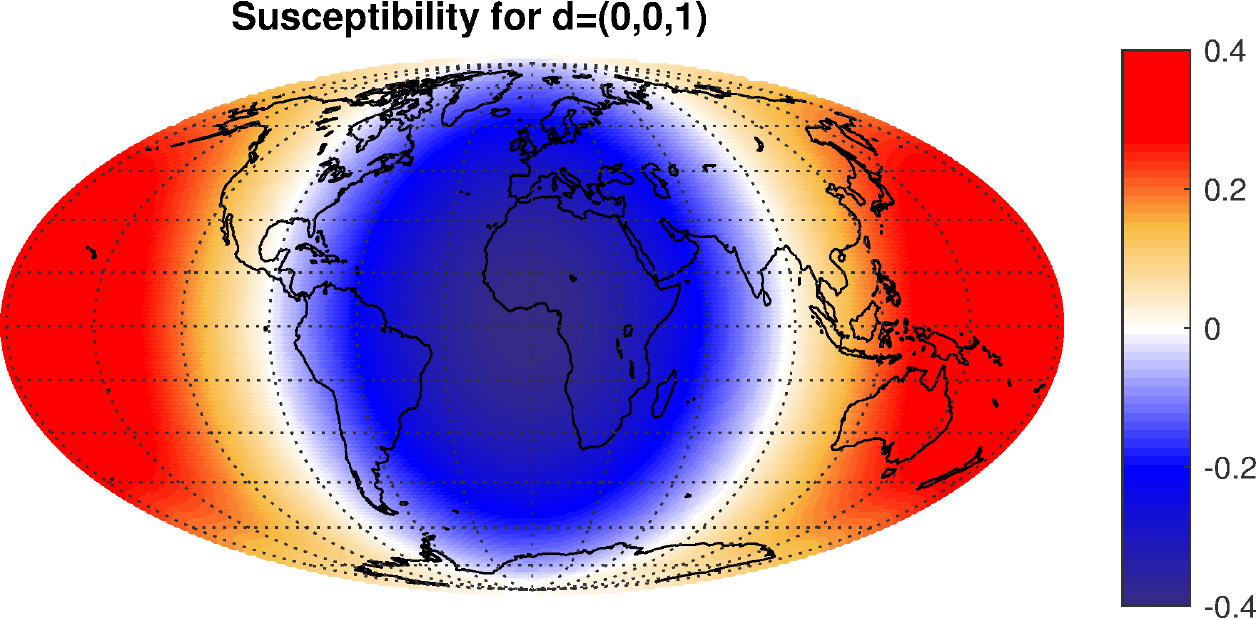}\quad\includegraphics[scale=0.46]{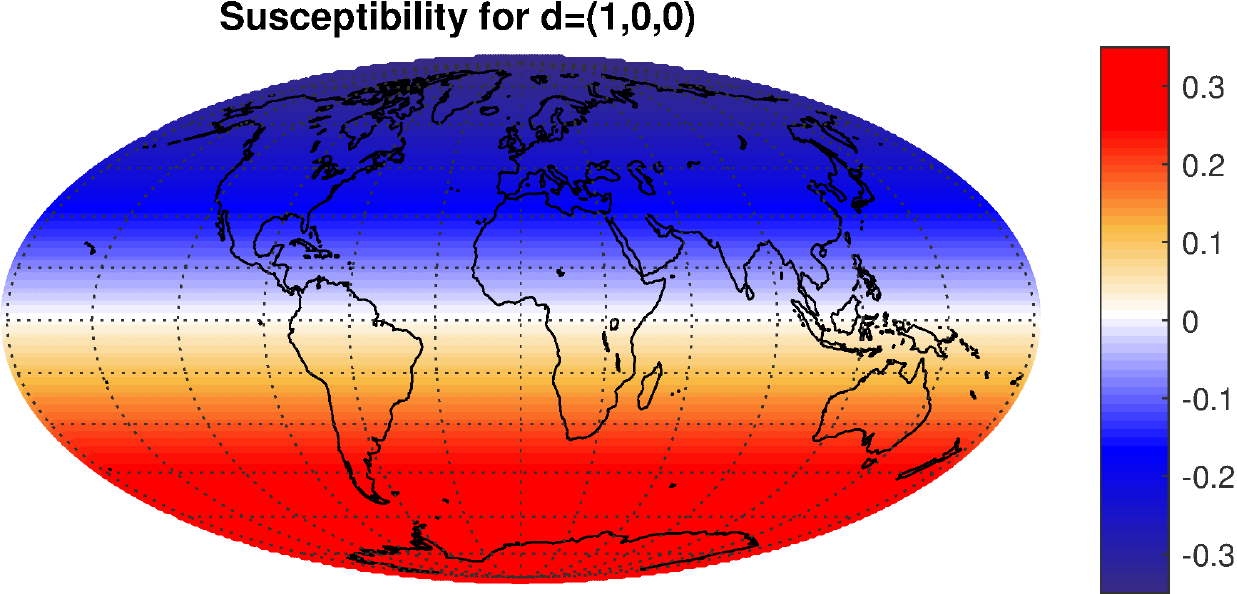}\quad\includegraphics[scale=0.38]{dip_direct-eps-converted-to.pdf} 
 \end{center}
 \caption{Illustration of the two band-limited dipole induced magnetizations with band-limit $N=1$ described Remark \ref{rem:suszbl} that are equivalent from outside. \emph{Left}: susceptibility $Q$, \emph{Center}: susceptibility $\overline{Q}$, \emph{Right}: dipole directions $\dd$ (blue) and $\overline{\dd}$ (red).}\label{fig:suszex2}
\end{figure}

To conclude this section, we summarize the previous considerations in the upcoming theorem. We actually formulate a slightly more general version that allows to decide if, for a given band-limited $\overline{\m}\in L^2(\SS,\R^3)$ (not necessarily of dipole induced form \eqref{eqn:m}), there exists a dipole direction $\dd\in\SS$ and a susceptibility $Q\in L^2(\SS)$ such that $(Q,\dd)$ and $\overline{\m}$ are equivalent from outside. This is essentially a band-limited counterpart to Theorem \ref{thm:reconst}.

\begin{thm}\label{thm:reconstbl}
Let $\overline{\m}\in L^2({\SS},\R^3)$ be band-limited with band-limit $N+1$.  Then there exists a band-limited susceptibility $Q\in L^2({\SS})$ and a dipole direction $\dd\in{\SS}$ such that $(Q,\dd)$ is equivalent to $\overline{\m}$ from outside if and only if there exists a vector $\y\in\C^3$ that is a zero of the vector-valued polynomial
\begin{align*}
\mathbf{T}_{\overline{\m}}=\mathbf{R}_{N+1,\overline{\m}}^2+\mathbf{R}_{N,\overline{\m}}^2+\sum_{p=2}^{N-1}\mathbf{S}_{p,\overline{\m}}^2,
\end{align*}
and that can be written in the form $\y=\y_\dd=(Y_{1,-1}({\dd}),Y_{1,0}({\dd}),Y_{1,1}({\dd}))$ for $\dd\in{\SS}$. The square in $\mathbf{R}_{p,\overline{\m}}^2$ and $\mathbf{S}_{p,\overline{\m}}^2$ is to be understood as acting componentwise on the vectors. The polynomials $\mathbf{R}_{p,\overline{\m}}$ and $\mathbf{S}_{p,\overline{\m}}$ are defined by
\begin{align*}
 \mathbf{R}_{p,\overline{\m}}(\mathbf{y})&=(\Pi_{q=-p+2}^p{c}_q(p,\y)\left(\overline{\m}_{p}-\mathbf{M}^{\y}_{p-1}\left(\widehat{\mathbf{M}}^{\y}_{p-1}\right)^{-1}\widehat{\overline{\m}}_{p}\right),
\\
 \mathbf{S}_{p,\overline{\m}}(\mathbf{y})&=(\Pi_{n=p}^{N}\Pi_{q=-n+2}^n{c}_q(n,\y))\bigg(\overline{\m}_{p}-\mathbf{N}^{\y}_{p+1}{\mathbf{q}}_{p+1}-\mathbf{M}^{\y}_{p-1}\left(\widehat{\mathbf{M}}^{\y}_{p-1}\right)^{-1}\left(\widehat{\overline{\m}}_{p}-\widehat{\mathbf{N}}^{\y}_{p+1}{\mathbf{q}}_{p+1}\right)\bigg),
\end{align*}\normalsize
with $\overline{\m}_p=\big((\tilde{\overline{m}}^{(2)})^\wedge(p,-p),\ldots, (\tilde{\overline{m}}^{(2)})^\wedge(p,p)\big)$ and $\mathbf{q}_{p}=(Q^\wedge(p,-p),\ldots,Q^\wedge(p,p))$. The matrices $\mathbf{M}^{\y}_{p-1}$ and $\mathbf{N}^{\y}_{p+1}$ are given as in \eqref{eqn:M} and \eqref{eqn:M2}. $\widehat{\mathbf{M}}^{\y}_{p-1}$ and $\widehat{\mathbf{N}}^{\y}_{p+1}$ denote the matrices ${\mathbf{M}}^{\y}_{p-1}$ and ${\mathbf{N}}^{\y}_{p+1}$, respectively, with its first two rows (for $\dd\in \SS\setminus\{(0,0,\pm1)\}$) or  its first and last row (for $\dd=(0,0,\pm1)$) deleted. Analogously, $\widehat{\overline{\m}}_{p}$ represents the vector ${\overline{\m}}_{p}$ with its first two entries or its first and last entry deleted.
The vectors $\mathbf{q}_p$, containing the Fourier coefficients of $Q$, can be computed iteratively by 
\begin{align*}
 \mathbf{q}_{p-1}&=\left(\widehat{\mathbf{M}}^{\y}_{p-1}\right)^{-1}\widehat{\overline{\m}}_{p},\quad p\in\{N,N+1\},
 \\\mathbf{q}_{p-1}&=\left(\widehat{\mathbf{M}}^{\y}_{p-1}\right)^{-1}\left(\widehat{\overline{\m}}_{p}-\widehat{\mathbf{N}}^{\y}_{p+1}{\mathbf{q}}_{p+1}\right),\quad p=N-1,\ldots,2.
\end{align*}
\end{thm}

\begin{proof} 
The condition \eqref{eqn:Msys2} for two dipole induced magnetizations can be rewritten in the following way
\begin{align}\label{eqn:Msys5}
 \mathbf{M}^\y_{p-1}\mathbf{q}_{p-1}+\mathbf{N}^\y_{p+1}\mathbf{q}_{p+1}-\overline{\m}_{p}=0,\quad p\geq1,
\end{align}
to fit the setup of the theorem. The desired results then follow in the exact same manner as described in the previous paragraphs. The polynomial $\mathbf{T}_{\overline{\m}}$ has only been introduced to obtain a single non-negative polynomial of which the zeros have to be found, rather than finding zeros separately for all $\mathbf{R}_{p,\overline{\m}}$ and $\mathbf{S}_{p,\overline{\m}}$.
\end{proof}

\begin{rem}\label{rem:proc2}
Just as mentioned in Remark \ref{rem:proc}, for a given magnetic potential $V$, one first has to find a general magnetization $\overline{\m}$ such that $V=V[\overline{\m}]$ on $\SS_R$. Afterwards one can use Theorem \ref{thm:reconstbl} to check whether $V$ can also be expressed in the form $V=V[Q,\dd]$. For the construction of the polynomial $\mathbf{T}_{\overline{\m}}$ in Theorem \ref{thm:reconstbl}, only the contribution $\tilde{\overline{\m}}^{(2)}$ of $\overline{\m}$, which is determined uniquely by $V$, is required.
\end{rem}

\section{Numerical Examples}\label{sec:num}

We now provide some numerical examples for the considerations in Section \ref{sec:unique}. Remark \ref{rem:proc} motivates the following two-step procedure to check whether a susceptibility $Q$ and dipole direction $\dd\in\SS$ exist such that $V[Q,\dd]={V}$ for a given potential ${V}$ on $\SS_R$ and to actually compute such $Q$, $\dd$. In fact, the focus is on finding a suitable dipole direction $\dd$ (this is the quantity of interest, e.g., in some paleomagnetic problems; and once the dipole direction is known, the susceptibility could be obtained by solving the linear inverse problem $V[Q,\dd]={V}$ for a given $\dd$).

\begin{meth}\label{meth:1}
 Let a magnetic potential ${V}$ be given on a sphere $\SS_R$ of radius $R>1$, and let $\Gamma\subset\SS$ be a subregion with $\Gamma\not=\SS$. Then proceed as follows:
 \begin{enumerate}
  \item[(1)] Find a magnetization $\m^*\in H_1({\SS},\R^3)$ with $\textnormal{supp}(\m^*)\subset\Gamma$ that satisfies 
  \begin{align*}
   V[\m^*](x)={V}(x), \quad x\in\SS_R.
  \end{align*}
  By $V[\m]$ we denote the magnetic potential generated by $\m$ via \eqref{eqn:smag0}.
  \item[(2)] Compute $g^*$, $\mathbf{h}^*$, and $m_1^*$ from the $\m^*$ obtained in (1). Find a $\dd^*\in{\SS}$ that satisfies
  \begin{align*}
   T_{g^*,\mathbf{h}^*,m_1^*}(\dd^*)=\min_{\dd\in{\SS}} T_{g^*,\mathbf{h}^*,m_1^*}(\dd).
  \end{align*}
  \item[(3)] Find a susceptibility $Q^*\in H_1({\SS})$ with $\textnormal{supp}(Q^*)\subset\Gamma$  such that 
  \begin{align*}
   V[Q^*,\dd^*]= {V}(x), \quad x\in\SS_R.
  \end{align*}
  By $V[Q,\dd]$ we denote the magnetic potential $V[\m]$ generated by a magnetization $\m$ of the form \eqref{eqn:m}. 
  \end{enumerate}
  If the data misfit $\|V[Q^*,\dd^*]-{V}\|_{L^2(\SS_R)}$ in step (3) is 'too large', go back to (2), find a new $\dd^{**}\not=\dd^*\in\SS$ and repeat step (3) with this $\dd^{**}$. If no other $\dd^{**}$ exists, this is an indicator that the given magnetic potential ${V}$ cannot be produced by a dipole induced magnetization. If $\|V[Q^*,\dd^*]-{V}\|_{L^2(\SS_R)}$ in step (3) is 'sufficiently small', then $Q^*$ and $\dd^*$  represent a susceptibility and a dipole direction with the desired properties.
\end{meth}

\begin{rem}
Concerning step (2) in Procedure \ref{meth:1}, Theorem \ref{thm:reconst} actually requires to find a zero $\dd^*$ of $T_{g^*,\mathbf{h}^*,m_1^*}$. However, such a zero might not exist either because there does not exist a dipole induced magnetization that produces ${V}$ in the first place or because noise in the measurements or reconstruction errors may have lead to a deteriorated version of  $T_{g^*,\mathbf{h}^*,m_1^*}$. In order to exclude false conclusions due to latter mentioned error sources, we minimize  $T_{g^*,\mathbf{h}^*,m_1^*}$ instead of trying to find its zeros (since  $T_{g^*,\mathbf{h}^*,m_1^*}$ is always non-negative by construction, this procedure is justified). If  $T_{g^*,\mathbf{h}^*,m_1^*}(\dd^*)$ is 'too large', this is an indicator that no zero exists and, thus, no dipole induced magnetization exists that produces ${V}$. The question of what 'too large' means is of course a delicate one, we illustrate it by some examples later on.

Since we are mainly interested in the dipole direction $\dd$, the first two steps in Procedure \ref{meth:1} are the important ones. But step (3) can be seen as a validation of the result of the first two steps: Theorem \ref{thm:reconst} requires $P_{{\m}^*,\dd^*}\in H_1({\SS})$ in order to guarantee that there exists a $Q^*$ such that $V[Q^*,\dd^*]=V$ on $\SS_R$ (in that case, $Q^*(x)= \frac{x\cdot\m^*(x)}{2x\cdot \dd^*}$ would be the corresponding susceptibility). However, due to measurement and reconstruction errors in ${V}$ and $\m^*$, respectively, it is unlikely that $P_{{\m}^*,\dd^*}\in H_1({\SS})$ for the $\m^*$ and $\dd^*$ obtained in steps (1) and (2). Thus, it is reasonable to invert ${V}$ again in step (3), now with a given $\dd^*$, in order to obtain an approximation of $Q^*$ that lies in $H_1(\SS)$. If the data misfit $\|V[Q^*,\dd^*]-{V}\|_{L^2(\SS_R)}$ is 'small enough', this indicates that $\dd^*$ is an admissible dipole direction.
 
Last but not least, it should be noted that the inverse problems in step (1) and (3) of Procedure \ref{meth:1} are linear (opposed to computing approximations $Q^*$ and $\dd^*$ directly from a single inversion of ${V}$). Additionally, Procedure \ref{meth:1} supplies more information on possible candidates for dipole directions than the direct inversion, since it is fairly easy to find minimizers of $T_{g^*,\mathbf{h}^*,m_1^*}$ in step (2).
\end{rem}

We illustrate Procedure \ref{meth:1} for three different situations. All situations have in common that the potential ${V}$ is given on $\SS_R$, with $R=1.06$ (which simulates the situation of a satellite flying at an altitude of around $380$km above the Earth's surface). Furthermore, ${V}$ is assumed to be given only in discrete points on an equiangular grid of $40,401$ points. The magnetization $\m$ on ${\SS}$ that generates ${V}$ is varied among the three situations, but it is always supported in the lower hemisphere, i.e., $\textnormal{supp}(\m)\subset \Gamma=\{x\in{\SS}:x\cdot \mathbf{v}\leq0\}$ for $\mathbf{v}=(0,0,1)^T$ being fixed:
\begin{itemize}
 \item [(a)] $\m$ is a dipole induced magnetization that is uniquely determined. In particular,  $\m$ is of the form \eqref{eqn:m} with dipole direction $\dd=(0,0.436,0.9)^T$ and susceptibility
 \begin{align*}
  Q(x)=4(x\cdot \mathbf{v})^3\chi_{[-1,0]}(x\cdot \mathbf{v}),\quad x\in{\SS},
 \end{align*}
where $\chi_{[-1,0]}$ denotes the characteristic function on the interval $[-1,0]$.
\item[(a')] Same as in (a) but only a noisy version ${V}^\eps$ of ${V}$ is given. In this example, we choose the noise level $\eps={\|{V}^\eps-{V}\|_{L^2(\SS_R)}}/{\|{V}\|_{L^2(\SS_R)}}=10^{-2}$.
\item[(b)]$\m$ is a dipole induced magnetization that is non-unique and of a form as described in Remark \ref{rem:smooth}. In particular, we choose the dipole direction $\dd=(1,0,0)^T$ and the susceptibility 
 \begin{align*}
  Q(x)=P\left(x\cdot\frac{\dd\times{\dd}}{|\dd\times {\dd}|}\right)x\cdot {\dd},\quad x\in{\SS},\qquad P(t)=e^{-\frac{1}{t^2}}\chi_{[-1,0]}(t),\quad t\in[-1,1],
 \end{align*}
 where ${\dd}=(0,1,0)^T$ is a fixed auxiliary vector. (According to Remark \ref{rem:smooth}, choosing ${Q}(x)=P\big(x\cdot(\dd\times{\dd})/|\dd\times {\dd}|\big)x\cdot {\dd}$ yields a further dipole induced magnetization that is equivalent to $\m$ from above. In other words, $({Q},{\dd})$ is equivalent from above to $(Q,\dd)$.)
 \item[(c)] $\m$ is not a dipole induced magnetization. In particular, we choose 
 \begin{align*}
  \m(x)=Q(x)\mathbf{v},\quad x\in{\SS},
 \end{align*}
with $Q$ as in (a).
\end{itemize}

For each of the situations above we apply the first two steps of Procedure \ref{meth:1} (the third step is only indicated for situation (a')). In step (1), we construct $\m^*$ to be the minimizer of the functional
\begin{align}\label{eqn:minfunc}
 \mathcal{F}[\m]=\left\|V[\m]-{V}\right\|_{L^2(\SS_R)}^2+\alpha\|\m\|_{H_1({\SS},\R^3)}^2+\beta\|\m\|_{L^2(\SS\setminus\Gamma,\R^3)}^2,
\end{align}
where $\|\m\|_{H_1({\SS},\R^3)}$ denotes the Sobolev norm $\m$ (see, e.g., \cite{freedenschreiner09} for more details). 
The first term in \eqref{eqn:minfunc} simply represents a data misfit that measures the deviation of $V[\m]$ from the known magnetic potential ${V}$, while the second term is a Tikhonov-type regularization to reduce noise amplification resulting from the ill-posedness of the downward continuation of the potential field data ${V}$ to the surface ${\SS}$ (this is well-studied and can be found, e.g., in \cite{freeden99, lu15} and references therein). The third term in \eqref{eqn:minfunc} eventually penalizes magnetizations $\m$ that have contributions outside $\Gamma$, i.e., magnetizations that do not satisfy supp$(\m)\subset\Gamma$. For the discretization of $\mathcal{F}[\m]$, we expand $\m$ in terms of (vectorial) Abel-Poisson kernels: 
\begin{align}
 \m( x)&=\sum_{i=1}^3\sum_{n=1}^N\gamma_{i,n} \,o^{(i)}_x K( x\cdot  x_n),
 \\K( x\cdot x_n)&=\frac{1-h^2}{(1+h^2-2h( x\cdot x_n))^{\frac{3}{2}}},
\end{align}
where $h\in(0,1)$ is a fixed parameter (influencing the localization of $K$; we use $h=0.9$) and $\{ x_n\}_{n=1,\ldots,N}\subset{\SS}$ is a set of uniformly distributed points indicating the centers of the kernel $K$ (in our case, we choose $N=10235$ different centers). Some general properties of the Abel-Poisson kernel $K$ can be found, e.g., in \cite{freeden98}. With this discretization, the minimization of $\mathcal{F}[\m]$ reduces to solving a set of linear equations with respect to the coefficients $\gamma_{i,n}$.
%
In step (2), we compute $T_{g^*,\mathbf{h}^*,m_1^*}$ from the $\m^*$ obtained in step (1) and find its minimizers. For the purpose of illustration, we simply plot $T_{g^*,\mathbf{h}^*,m_1^*}$ over the sphere to indicate where the minima $\dd^*\in{\SS}$ are located. Eventually, given $\dd^*\in{\SS}$, in step (3) we minimize a functional $\mathcal{G}$ similar to \eqref{eqn:minfunc} in order to obtain $Q^*$. More precisely, we minimize
\begin{align}\label{eqn:minfunc2}
 \mathcal{G}[Q]=\left\|V[Q,\dd^*]-{V}\right\|_{L^2(\SS_R)}^2+\alpha\|\m[Q]\|_{H_1({\SS},\R^3)}^2+\beta\|\m[Q]\|_{L^2({\SS}\setminus\Gamma,\R^3)}^2,
\end{align}
where $\m[Q]$ denotes the induced magnetization $\m[Q]( x)=Q(x)(3(x\cdot \dd^*) x-\dd^*)$, $ x\in{\SS}$. For the numerical evaluation, we proceed similarly as for \eqref{eqn:minfunc} by expanding $Q$ in terms of (scalar) Abel Poisson kernels and solving a corresponding system of linear equations (details for a similar problem can be found in \cite{gerhards16a}). Any numerical integrations necessary during the procedure are performed via the methods of \cite{driscoll94} (when the integration region comprises the entire sphere ${\SS}$ or $\SS_R$, respectively) and \cite{hesse12} (when the integration is only performed over a spherical cap ${\SS}\setminus\Gamma$).

\begin{figure}\begin{center}
 \includegraphics[scale=0.4]{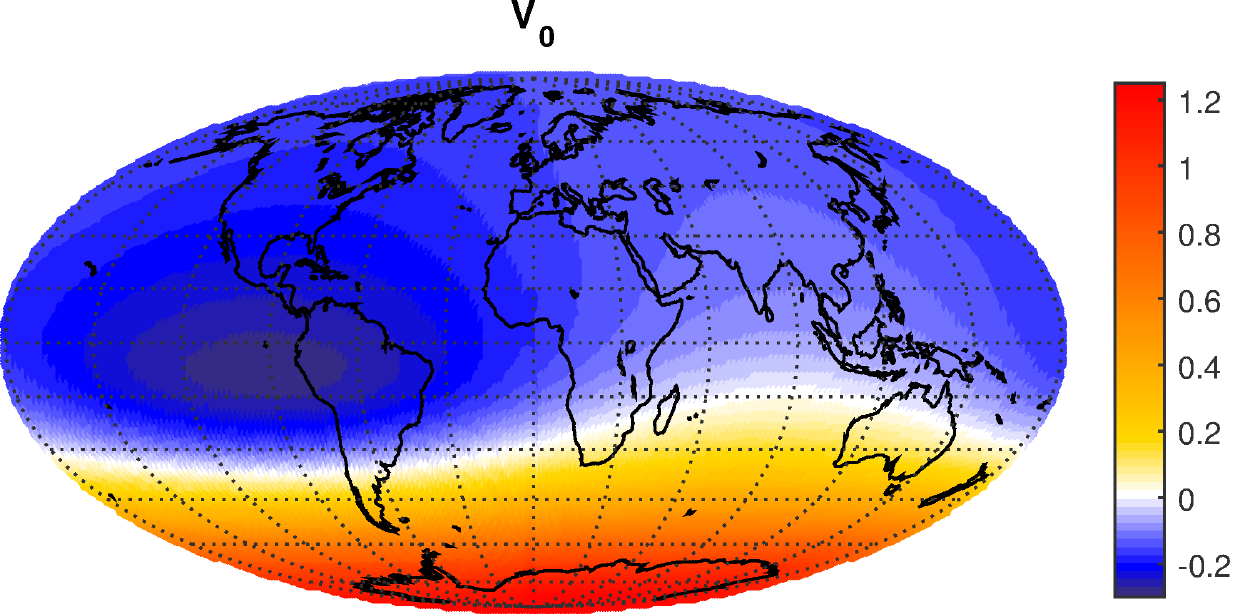}\quad\includegraphics[scale=0.42]{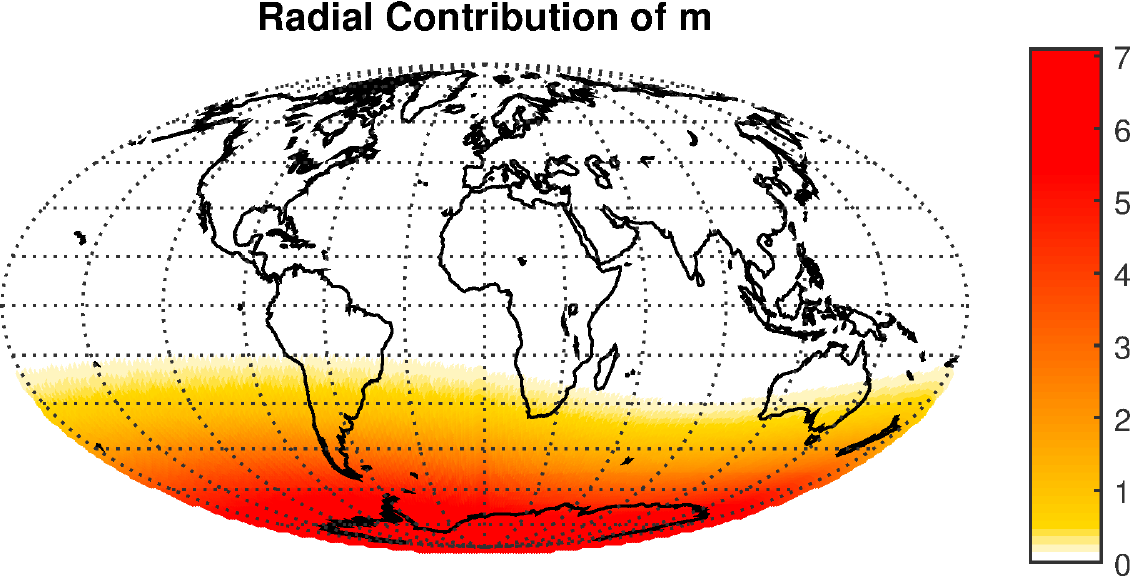}\quad\includegraphics[scale=0.42]{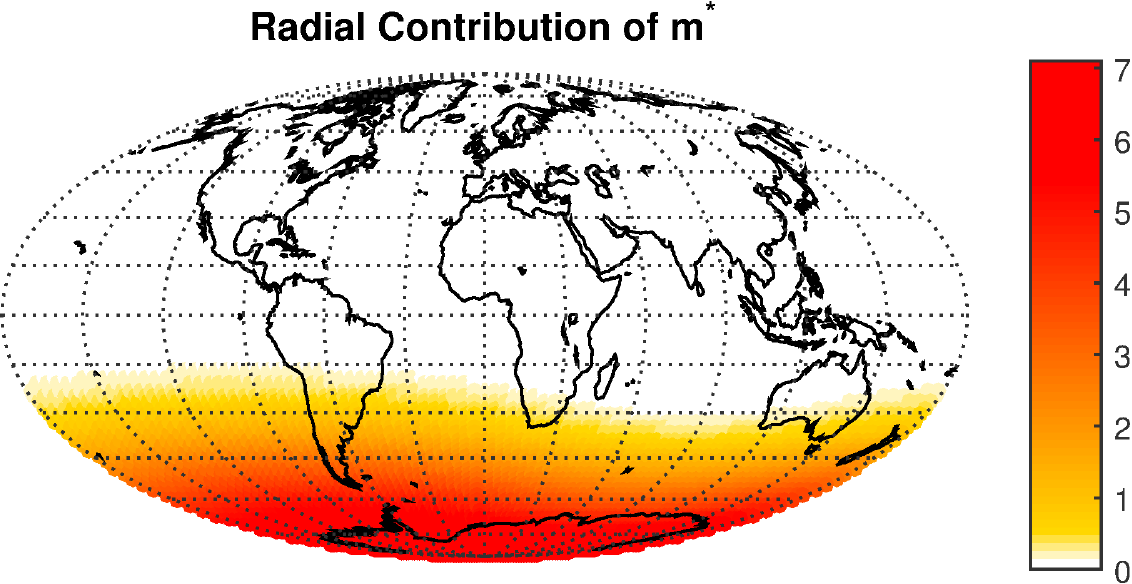}
 \end{center}
 \caption{Illustration of step (1) for situation (a): noise-free input data ${V}$ (\emph{left}), radial component $m_1$ of the true magnetization $\m$ (\emph{center}), and  radial component $m_1^*$ of the reconstructed magnetization $\m^*$ (\emph{right}).}\label{fig:case1a}
 \begin{center}
 \includegraphics[scale=0.5]{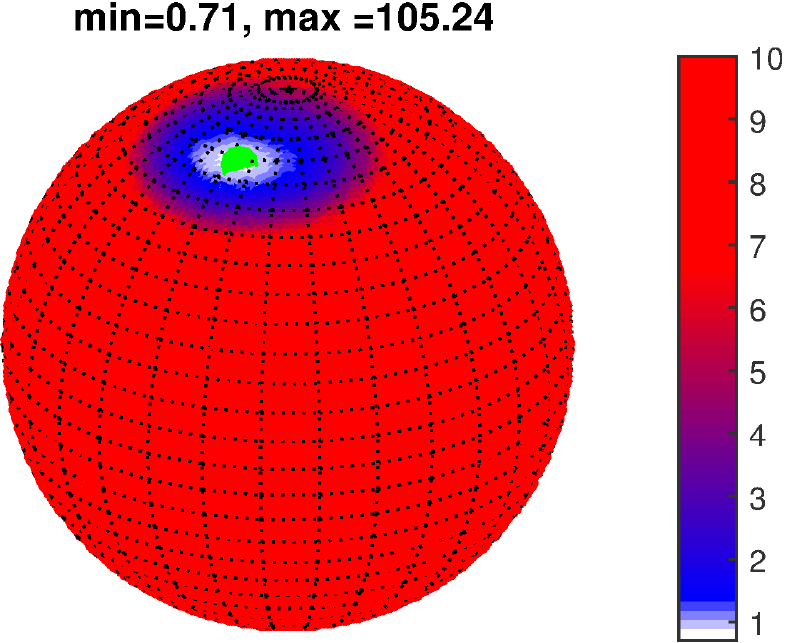} \end{center}
 \caption{Illustration of step (2) for situation (a): the figure shows the evaluation of $T_{g^*,\mathbf{h}^*,m_1^*}$ on the unit sphere, the \emph{green dot} indicates the location of the true dipole direction $\dd$. The color bar has been modified to emphasize the minimum, the actual minimum and maximum is indicated in the title.}\label{fig:case1b}
\end{figure}

\begin{figure}\begin{center}
 \includegraphics[scale=0.4]{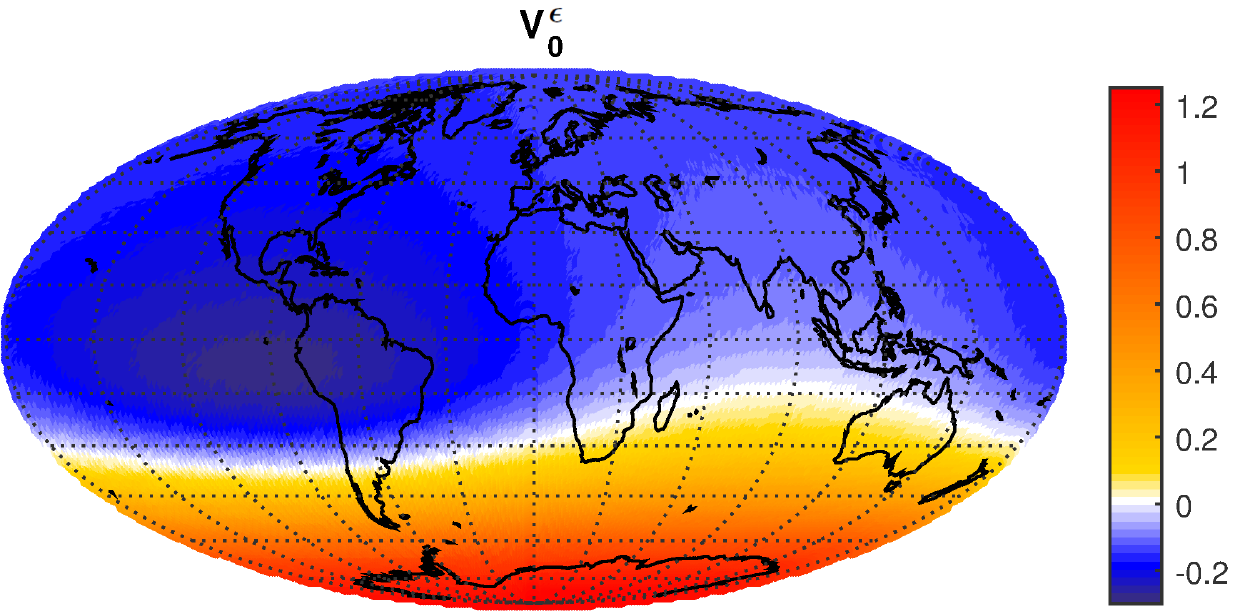}\quad\includegraphics[scale=0.42]{m_unique_no_noise-eps-converted-to.pdf}\quad\includegraphics[scale=0.42]{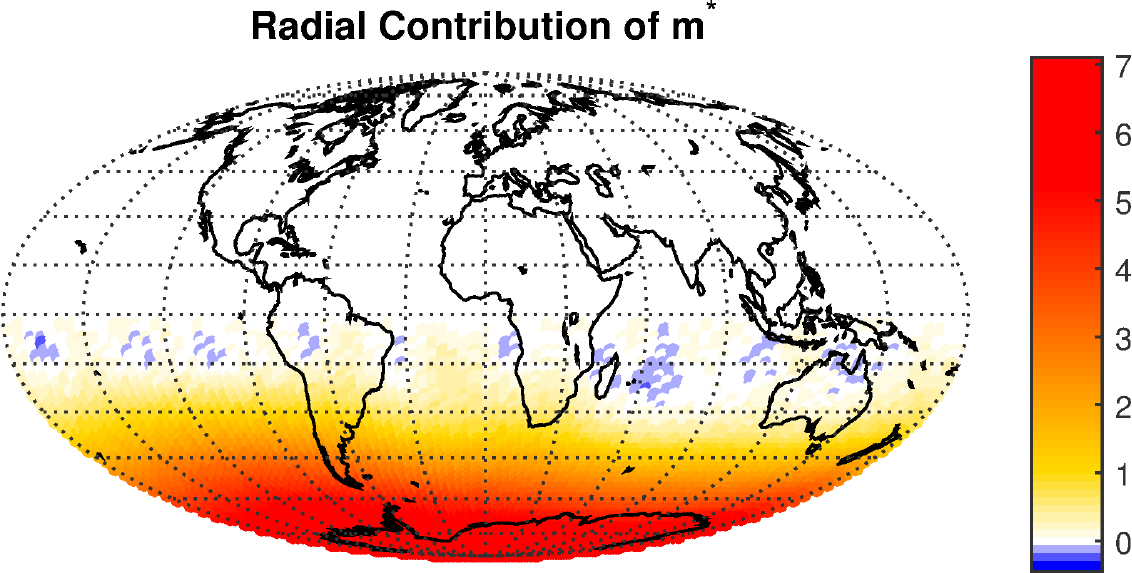}
 \end{center}
 \caption{Illustration of step (1) for situation (a'): noisy input data ${V}^\eps$ (\emph{left}), radial component $m_1$ of the true magnetization $\m$ (\emph{center}), and  radial component $m_1^*$ of the reconstructed magnetization $\m^*$ (\emph{right}).}\label{fig:case1'a}
 \begin{center}
 \includegraphics[scale=0.5]{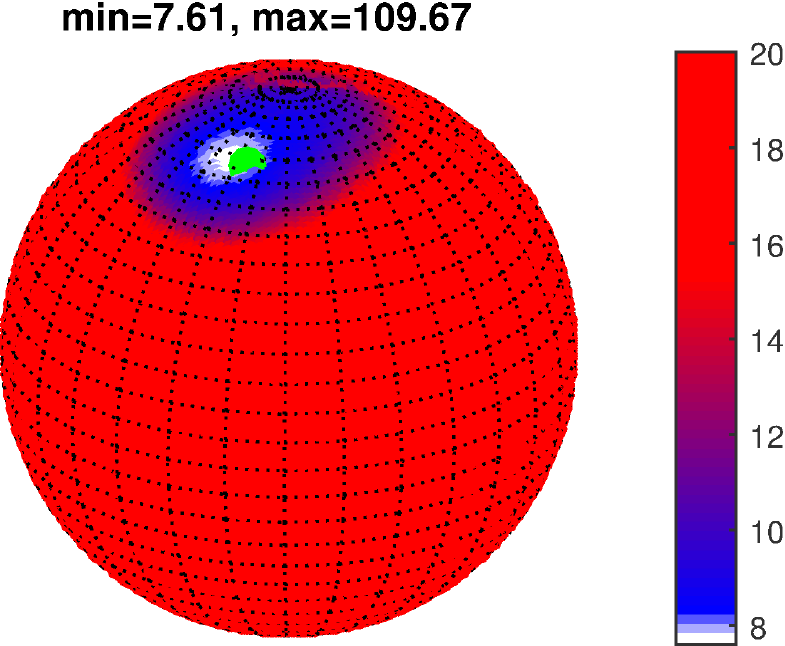} \end{center}
 \caption{Illustration of step (2) for situation (a'): the figure shows the evaluation of $T_{g^*,\mathbf{h}^*,m_1^*}$ on the unit sphere, the \emph{green dot} indicates the location of the true dipole direction $\dd$. The color bar has been modified to emphasize the minimum, the actual minimum and maximum is indicated in the title.}\label{fig:case1'b}
 \begin{center}
  \includegraphics[scale=0.4]{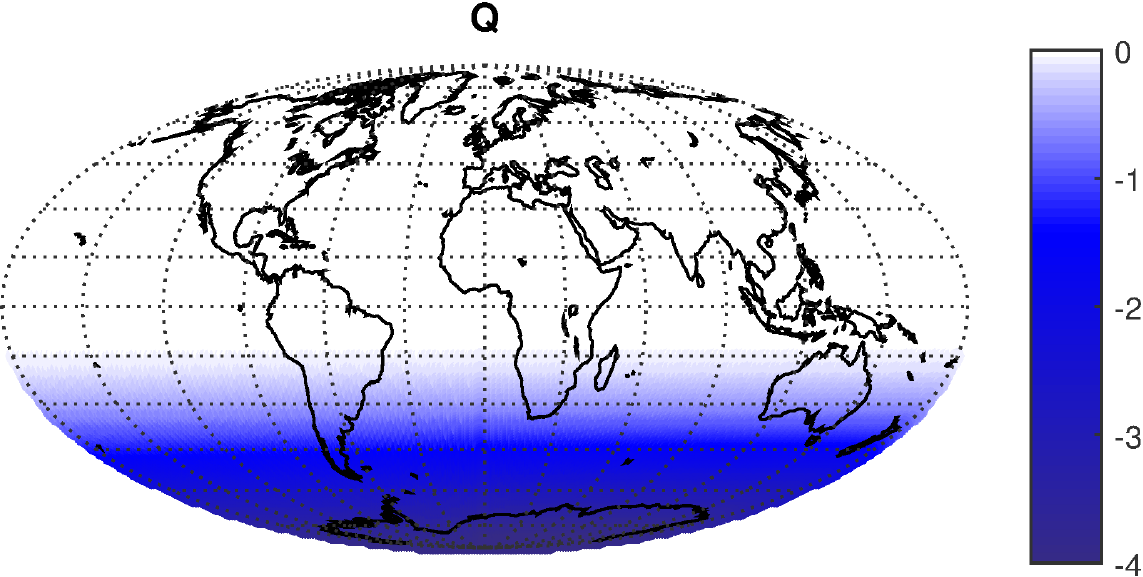}\quad\includegraphics[scale=0.42]{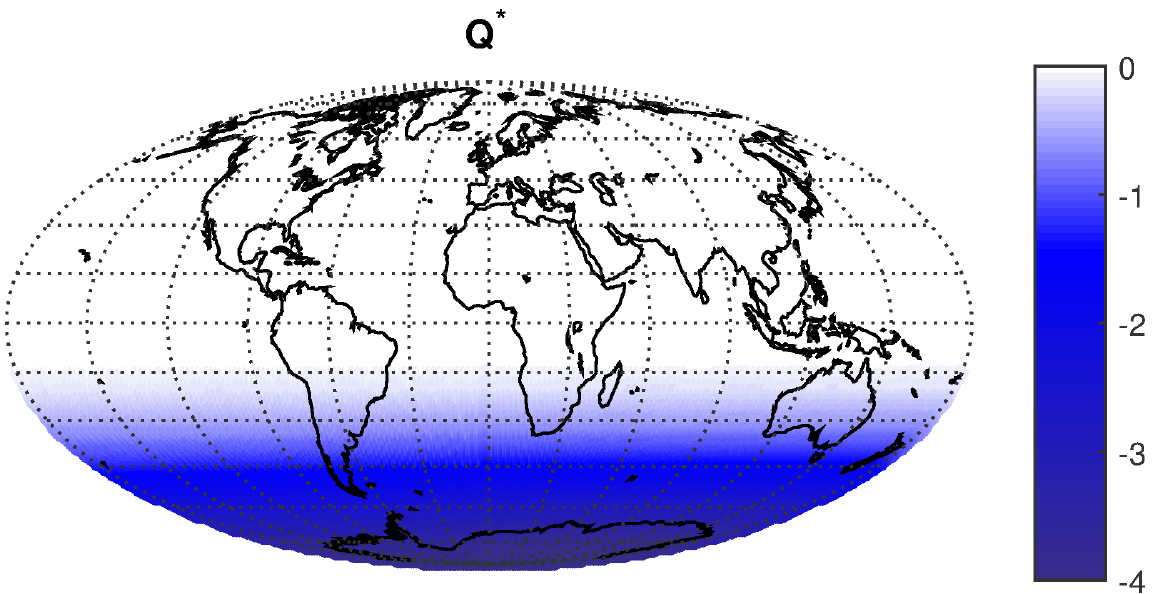}\quad\includegraphics[scale=0.42]{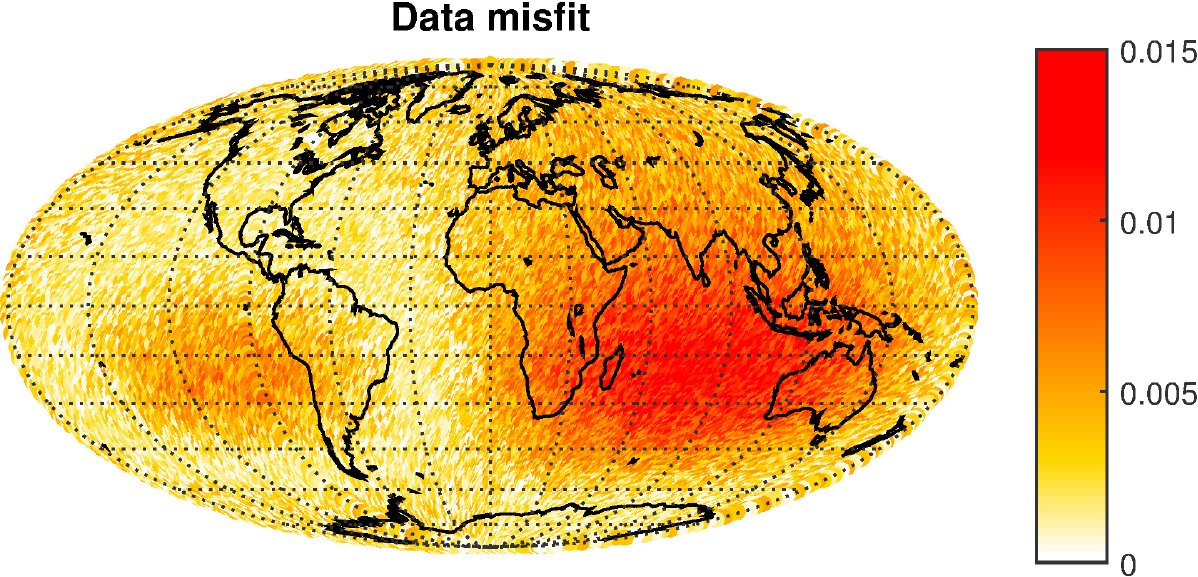} \end{center}
 \caption{Illustration of step (3) for situation (a'): true susceptibility $Q$ (\emph{left}) and reconstructed susceptibility $Q^*$ for  $\dd^*=(0.027,0.433,0.901)^T$ (\emph{center}). The \emph{data misfit} $|{V}^\eps-V[Q^*,\dd^*]|$ is indicated in the \emph{right} image.}\label{fig:case1'c}
\end{figure}

\begin{figure}\begin{center}
 $\begin{array}{ccc}\includegraphics[scale=0.4]{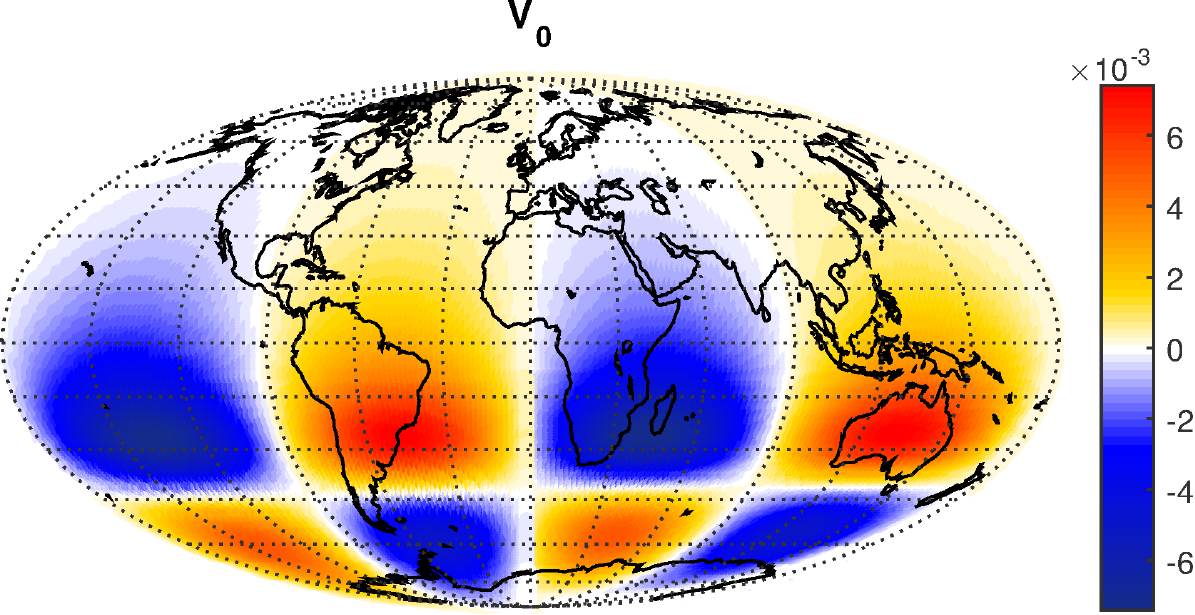}&\includegraphics[scale=0.42]{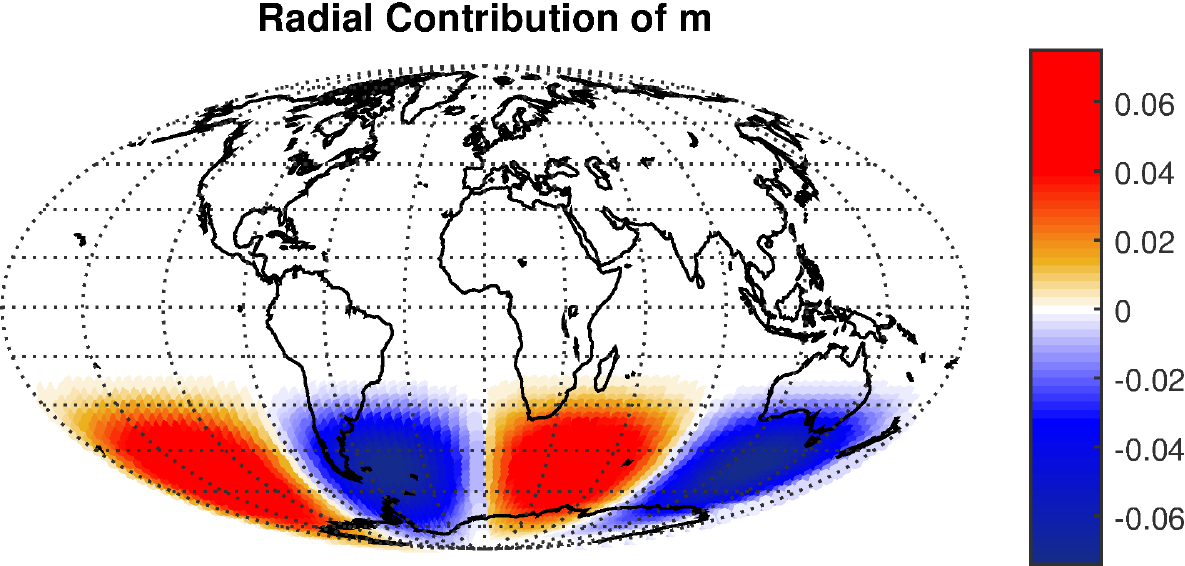}&\includegraphics[scale=0.42]{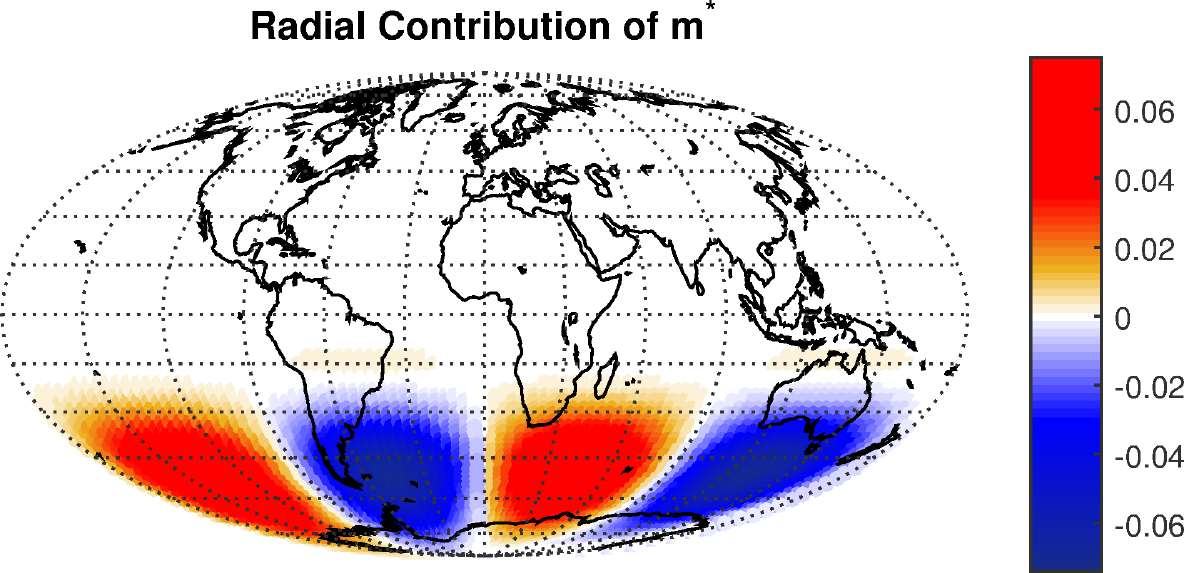}
 \\&\includegraphics[scale=0.42]{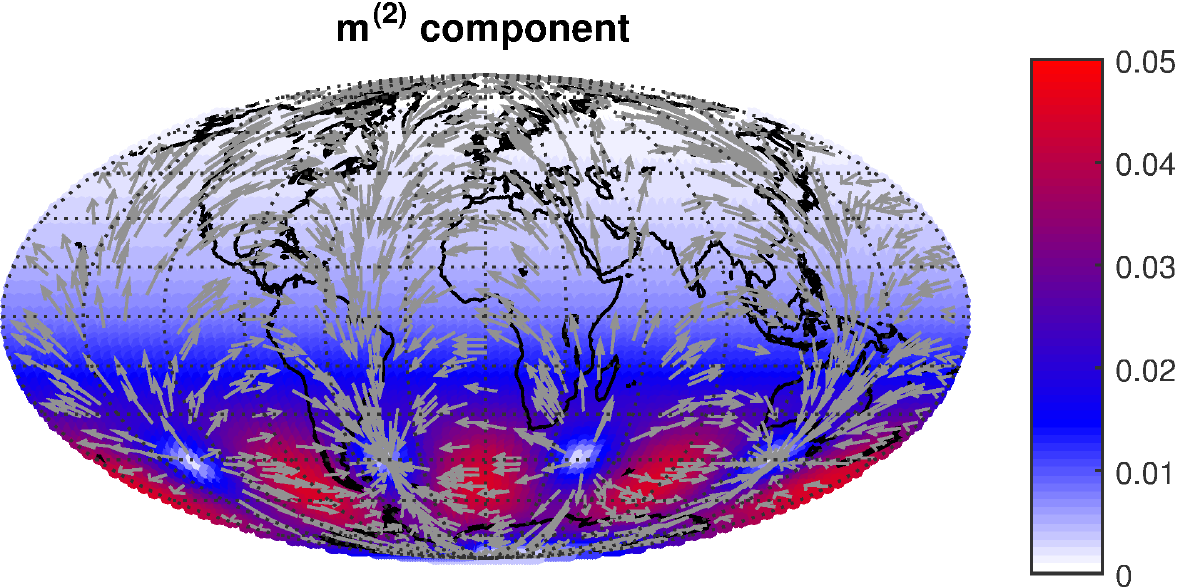}&\includegraphics[scale=0.42]{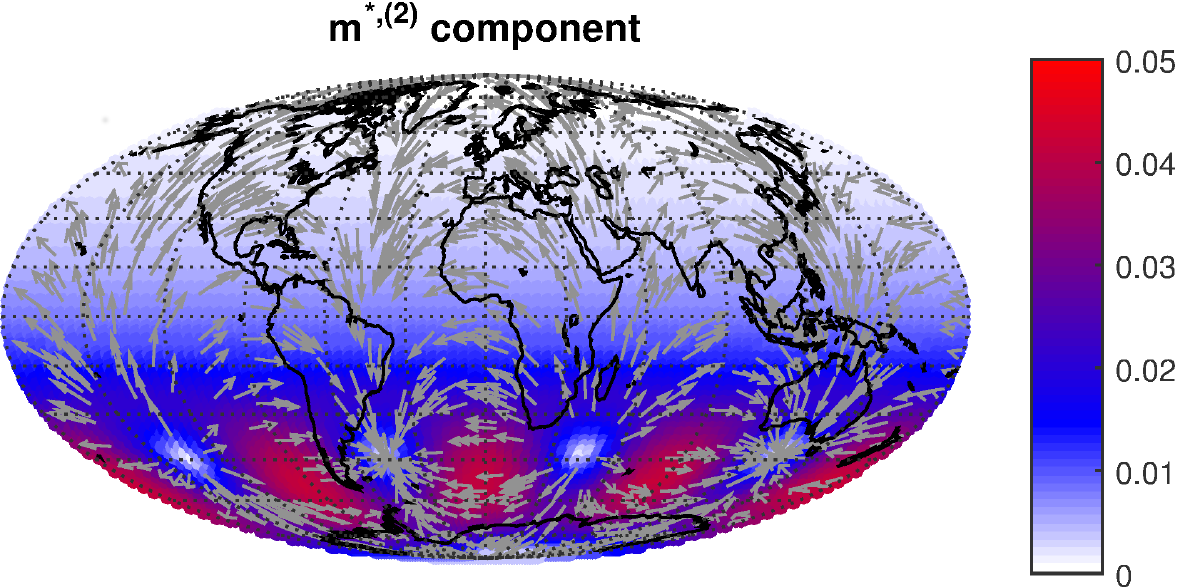}
 \\&\includegraphics[scale=0.42]{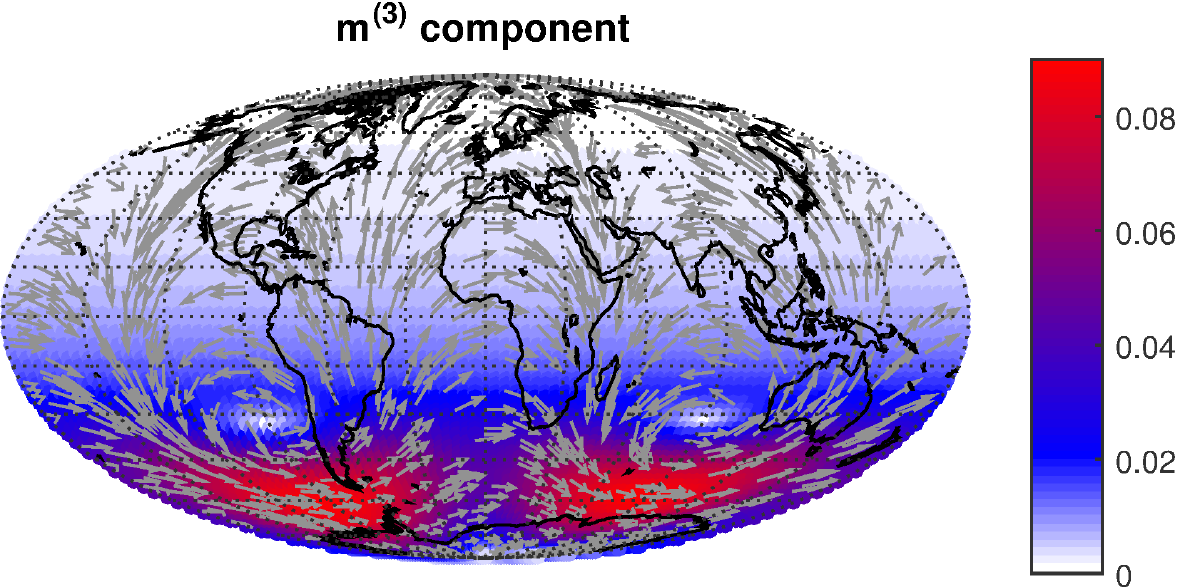}&\includegraphics[scale=0.42]{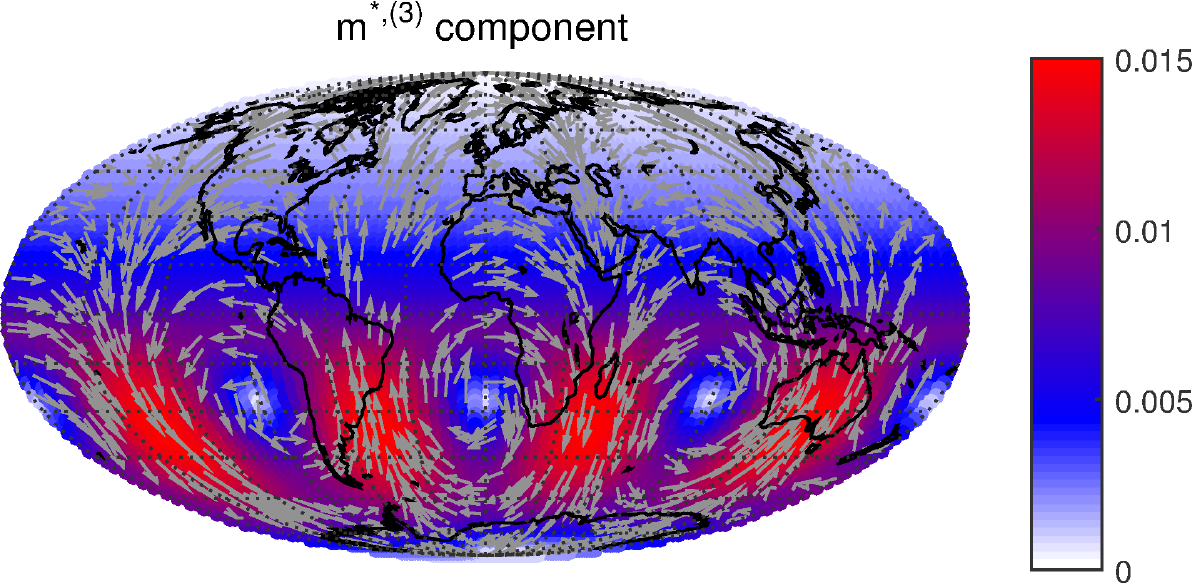}\end{array}$
 \end{center}
 \caption{Illustration of step (1) for situation (b): input data ${V}$ (\emph{left}), radial component $m_1$ and the contributions $\m^{(2)}$ and $\m^{(3)}$ of the true magnetization $\m$ (\emph{center}), and radial component $m_1^*$ and the contributions $\m^{*,(2)}$ and $\m^{*,(3)}$  of the reconstructed magnetization $\m^*$ (\emph{right}). In the plots of the second and third row, colors indicate the absolute values $|\m^{(i)}|$ and $|\m^{*,(i)}|$, $i=1,2$, and arrows the orientation.}\label{fig:case2a}
 \begin{center}
 \includegraphics[scale=0.4]{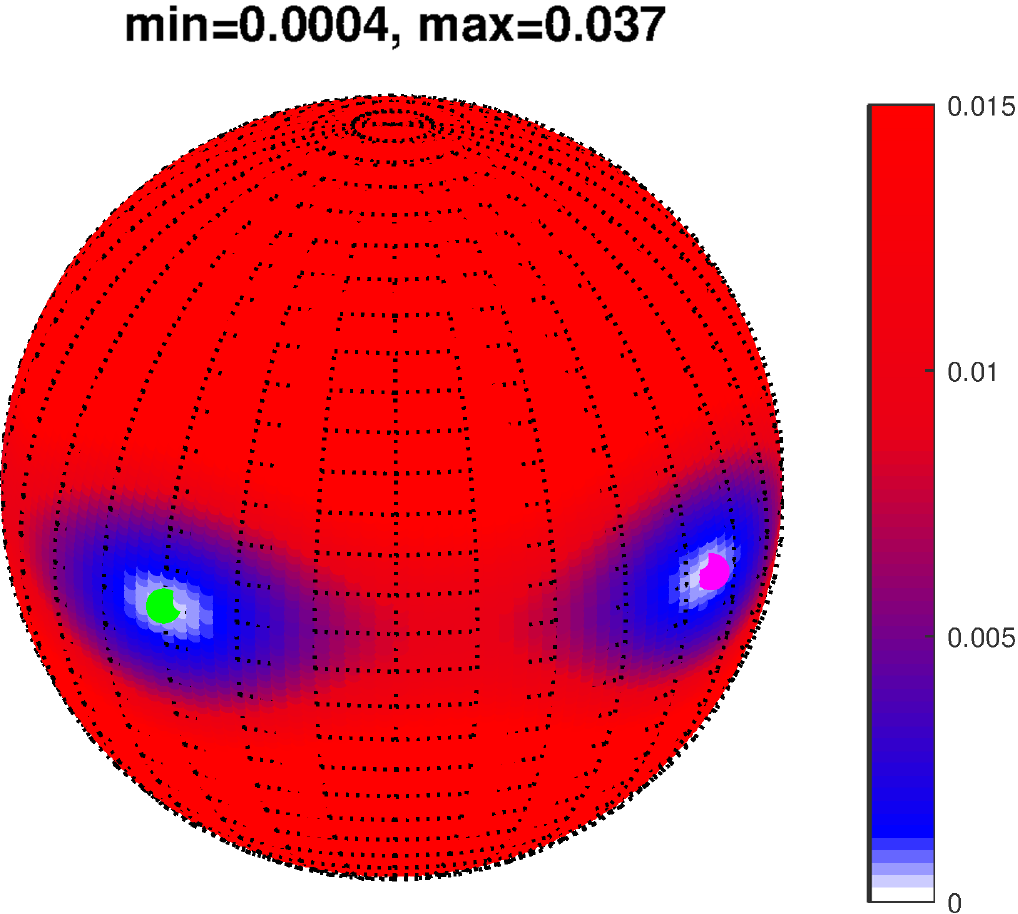} \end{center}
 \caption{Illustration of step (2) for situation (b): the figure shows the evaluation of $T_{g^*,\mathbf{h}^*,m_1^*}$ on the unit sphere, the \emph{green and purple dots} indicate the locations of the possible (true) dipole directions $\dd$ and ${\dd}$, respectively. The color bar has been modified to emphasize the minimum, the actual minimum and maximum is indicated in the title.}\label{fig:case2b}
\end{figure}

The results of step (1) and (2) for situation (a) are indicated in Figures \ref{fig:case1a} and \ref{fig:case1b}, respectively. The reconstruction $\m^*$ nicely fits the true $\m$. For brevity, we illustrated only the radial components. Figure \ref{fig:case1b} shows that the minimum of $T_{g^*,\mathbf{h}^*,m_1^*}$ coincides with the desired dipole direction $\dd$. The corresponding results for the noisy situation (a') are indicated in Figures \ref{fig:case1'a} and \ref{fig:case1'b}. We see that the reconstructed radial contribution of $\m^*$ shows some minor artifacts but the dipole direction $\dd$ still coincides quite well with the minimum of $T_{g^*,\mathbf{h}^*,m_1^*}$. In the perfect case it should hold that  $T_{g^*,\mathbf{h}^*,m_1^*}(\dd^*)=0$, however, we see that the actual minimum value is rather large in the noisy setup. Therefore, to make sure that we found a good candidate $\dd^*$ for the dipole direction, we proceed to step (3) with the approximation $\dd^*=(0.027,0.433,0.901)^T$ of the minimum of $T_{g^*,\mathbf{h}^*,m_1^*}$. The reconstructed susceptibility $Q^*$ and the true susceptibility $Q$ are indicated in Figure \ref{fig:case1'c} and they match very well, indicating that $\dd^*$ is a good approximation of the true dipole direction. The data misfit $|{V}^\eps-V[Q^*,\dd^*]|$ offers a decision criterion that does not require the knowledge of the true $Q$ and is also indicated in Figure \ref{fig:case1'c}. In this case, we see that the data misfit is small and we accept $\dd^*$ as an approximation of the true dipole direction.

Steps (1) and (2) for situation (b), where no uniqueness of $Q$ and $\dd$ is given, are shown in Figures \ref{fig:case2a} and \ref{fig:case2b}, respectively. In Figure \ref{fig:case2a} we indicated all three contributions (i.e., the radial contribution $m_1$ and the surface curl- and surface divergence-free contributions $\m^{(2)}$ and $\m^{(3)}$, respectively) of $\m$ and $\m^*$. It is seen that the radial contribution and the surface curl-free contribution of the true and the reconstructed magnetization coincide, as is expected from Theorem \ref{thm:sunique1}. However, the surface divergence-free contribution is not uniquely determined and therefore may differ, as is the case here. But latter has no impact on our further procedure. Figure \ref{fig:case2b} shows that the two possible dipole directions $\dd$ and ${\dd}$ are precisely the minima of $T_{g^*,\mathbf{h}^*,m_1^*}$. Which direction is the correct one cannot be decided without further a priori geophysical information due to the intrinsic non-uniqueness.

For situation (c), the magnetization $\m$ has been reconstructed very well as can be exemplarily seen for the radial component in Figure  \ref{fig:case3a}. Figure \ref{fig:case3b} shows that the minima of $T_{g^*,\mathbf{h}^*,m_1^*}$ are located on the equator, i.e., any possible candidate for a dipole direction $\dd^*$ must lie in the equatorial plane. However, the acquired minimum value is so large that this leads us to conclude that the potential ${V}$ cannot be generated by a dipole induced magnetization.

\begin{figure}\begin{center}
 \includegraphics[scale=0.38]{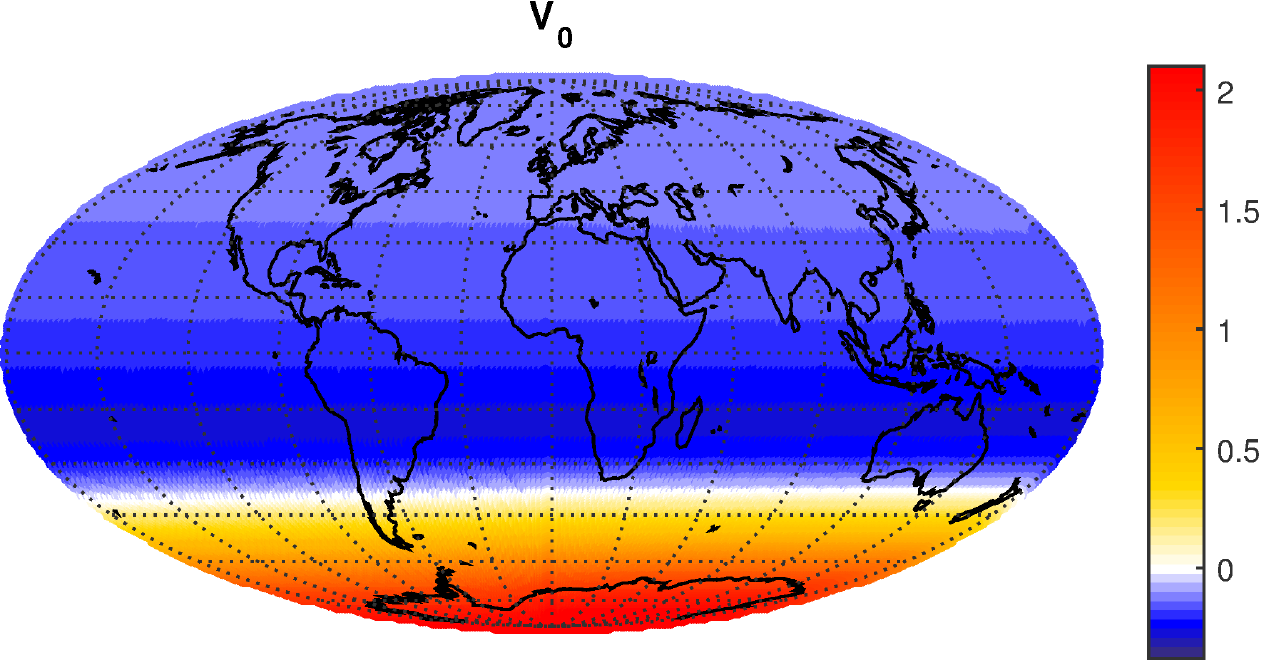}\quad\includegraphics[scale=0.42]{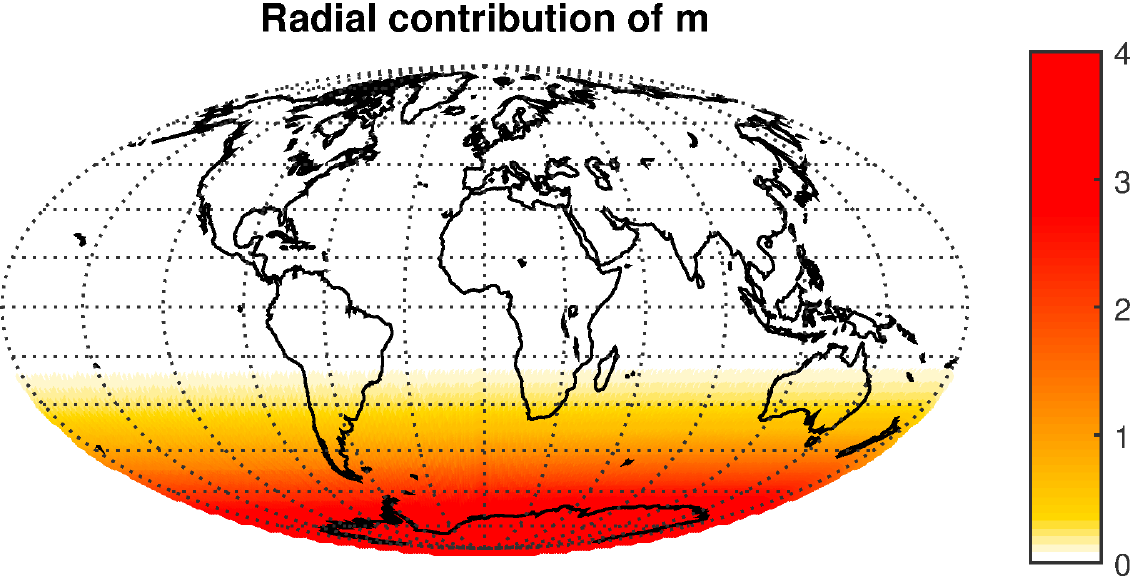}\quad\includegraphics[scale=0.42]{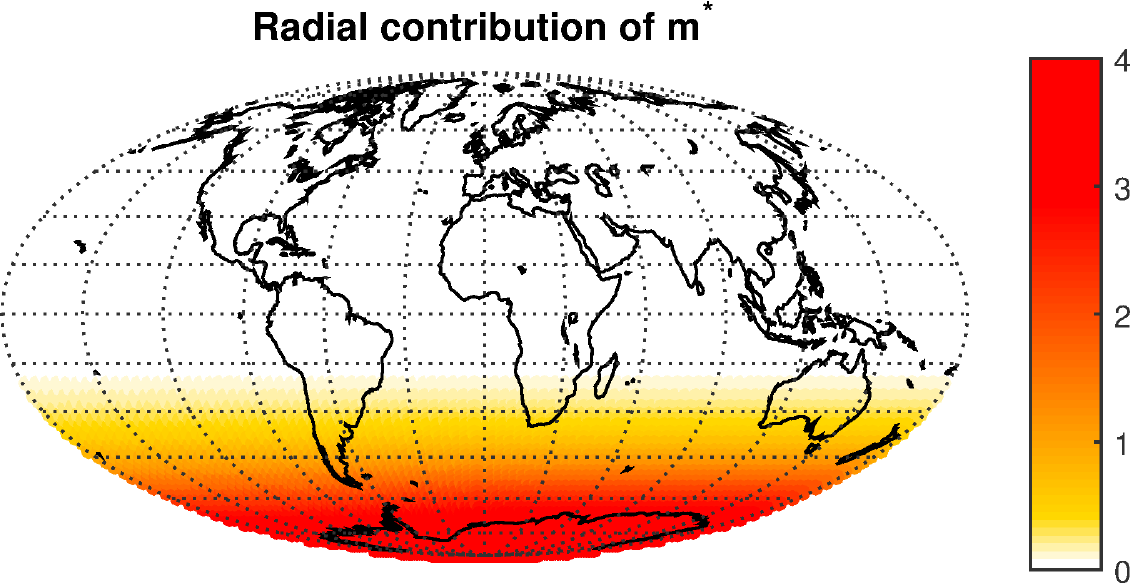}
 \end{center}
 \caption{Illustration of step (1) for situation (c): input data ${V}$ (\emph{left}), radial component $m_1$ of the true magnetization $\m$ (\emph{center}), and  radial component $m_1^*$ of the reconstructed magnetization $\m^*$ (\emph{right}).}\label{fig:case3a}
 \begin{center}
 \includegraphics[scale=0.6]{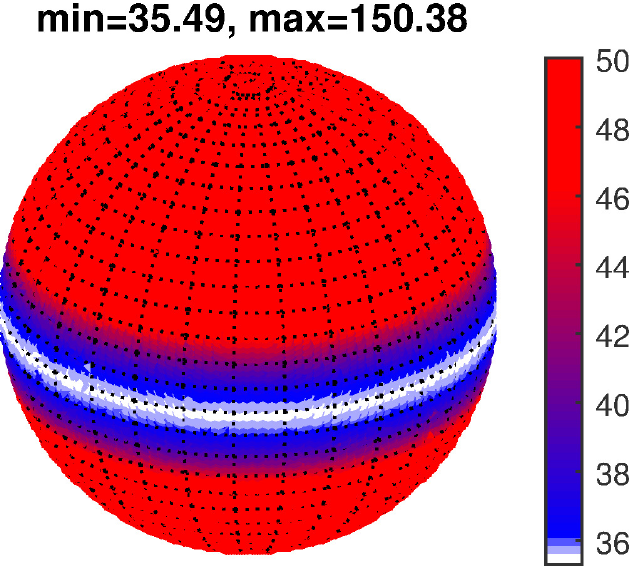} \end{center}
 \caption{Illustration of step (2) for situation (c): the figure shows the evaluation of $T_{g^*,\mathbf{h}^*,m_1^*}$ on the unit sphere. The color bar has been modified to emphasize the minimum, the actual minimum and maximum is indicated in the title.}\label{fig:case3b}
\end{figure}


\section{Conclusion}
The fact that generally only the $\tilde{\m}^{(2)}$-contribution of a spherical magnetization $\m$ can be uniquely reconstructed from satellite magnetic field measurements leads to uniqueness issues, e.g., in determining possible dipole directions (assuming that the underlying magnetization is of induced type). The additional assumption that $\m$ is localized in some subregion of a spherical planetary surface allows to uniquely determine the $\tilde{\m}^{(1)}$- and  $\tilde{\m}^{(2)}$-contributions of $\m$ (although $\tilde{\m}^{(3)}$ is still unknown), which implies that the radial contribution $\m^{(1)}$ and the tangential surface curl-free contribution $\m^{(2)}$ are determined uniquely. Here, we have shown that for the latter situation there exists a procedure for the determination of candidates for the dipole direction $\dd$ and for the decision if a measured magnetic field can be produced by a dipole induced magnetization in the first place (a similar procedure has been derived for band-limited magnetizations, but in our examples in Section \ref{sec:num} we focused on the spatial localization constraint as we believe it to be more feasible for actual applications). The numerical treatment of the involved extremal problems allows various approaches and should be investigated in more detail for future applications. The focus of this paper is on the presentation and illustration of the conceptual setup for the improved reconstruction of dipole directions and the investigation of uniqueness issues.
\\[2ex]
\begin{ackno}
 The author thanks Foteini Vervelidou, GFZ Potsdam, for pointing out the problem of studying the reconstruction of dipole directions. The work was partly supported by DFG grant GE 2781/1-1.
\end{ackno}

\footnotesize
\bibliographystyle{unsrt}

\end{document}